\documentclass{article} 
\usepackage{amsmath, amssymb, amsthm}
\usepackage{graphicx}
\usepackage{color}
\usepackage[top=30truemm,bottom=30truemm,left=25truemm,right=25truemm]{geometry}
\newtheorem{theo}{Theorem}
\newtheorem{prop}{Proposition}
\newtheorem{defi}{Definition}
\newtheorem{lemm}{Lemma}

\begin{document}

\author{
Yusuke Kinoshita
\thanks{
  Graduate School of Informatics, Nagoya University, Nagoya, Aichi, Japan, kinoshita.yusuke@j.mbox.nagoya-u.ac.jp
}
}
\title{Analysis of Quantum Multi-Prover Zero-Knowledge Systems: Elimination of the Honest Condition and Computational Zero-Knowledge Systems for $\sf QMIP^*$}
\maketitle
\begin{abstract}
Zero-knowledge and multi-prover systems are both central notions in classical and quantum complexity theory. There is, however, little research in quantum multi-prover zero-knowledge systems. This paper studies complexity-theoretical aspects of the quantum multi-prover zero-knowledge systems. This paper has two results:
\begin{itemize}
\item$\mathsf{QMIP^*}$ systems with honest zero-knowledge can be converted into general zero-knowledge systems without any assumptions.
\item$\mathsf{QMIP^*}$ has computational quantum zero-knowledge systems if a natural computational conjecture holds. 
\end{itemize}
One of the main tools is a test (called the GHZ test) that uses GHZ states shared by the provers, which prevents the verifier's attack in the above two results.
Another main tool is what we call  the Local Hamiltonian based Interactive protocol (LHI protocol). The LHI protocol makes previous research for Local Hamiltonians applicable to check the history state of interactive proofs, and we then apply Broadbent et al.'s zero-knowledge protocol for QMA \cite{BJSW} to quantum multi-prover systems in order to obtain the second result.
\end{abstract}
\thispagestyle{empty}
\newpage
\setcounter{page}{1}
\section{Introduction}
\subsection{Background}
Zero-knowledge systems are interactive proof systems that the prover (with unlimited computational power) can convince the verifier (with polynomial-time computational power) of the correctness of the claim without leakage of information. By ``without leakage'', we mean that there exists a polynomial-time simulator whose output is indistinguishable from the communication of the interactive proof system. Depending on the definition of the indistinguishability, there are several variants of zero-knowledge: perfect/statistical/computational zero-knowledge, that is, the output of the simulator is identical/statistically close/computational indistinguishable to that of the interactive proof system. It has been a central notion since it was proposed by Goldwasser, Micali, and Rackoff \cite{GMR}. Though it seems impossible, Graph Isomorphism \cite{GMW} and Non Quadratic Residue \cite{GMR} have perfect zero-knowledge systems, for example. 

Multi-prover interactive proof systems are introduced in the relation to zero-knowledge, which is also a central notion in computational theory and cryptography. Multi-prover interactive proof systems are interactive proof systems with many provers, who cannot communicate with each other during the protocols \cite{Bab,BGKW,GMR}. Each prover has to reply to the verifier without the information about communications between the other provers and the verifier, and this restricts malicious strategies by the provers. The class of languages that have multi-prover interactive proofs is denoted as $\mathsf{MIP}$. The computational equivalence $\mathsf{MIP=NEXP}$ was proved by Babai, Fortnow and Lund \cite{BFL}, and hence $\mathsf{IP}$, which is the class of languages that have single-prover interactive proofs ($\mathsf{=PSPACE}$  \cite{Sh}), and $\mathsf{MIP}$ are different unless $\mathsf{PSPACE=NEXP}$. Under the statistical or perfect zero-knowledge condition, multi-prover interactive proofs are much stronger than single-prover interactive proofs. Zero-knowledge proofs of NP with a single-prover need the assumption that one-way functions exist \cite{OW}, but multi-prover interactive proofs do not need any computational assumptions. Even multi-prover interactive proofs with perfect zero-knowledge can compute all languages in $\mathsf{MIP}$ \cite{BGKW}.

In quantum complexity theory, quantum analogues of these interactive proof systems have been also deeply studied. Quantum interactive proof system, where the prover and the verifier use quantum communication and the verifier can do polynomial-time quantum computations, were introduced by Watrous \cite{W0}. Quantum statistical zero-knowledge systems were also introduced by Watrous \cite{W1}. An important notion of zero-knowledge systems is the honest zero-knowledge, that is, an interactive proof system is honest zero-knowledge if it leaks no information only to the honest verifier who follows the specified protocol. Restricting the verifier to be honest makes the construction of the simulator easier, but the equivalence between general zero-knowledge and honest zero-knowledge is not obvious. This equivalence was proved by Watrous for statistical zero-knowledge by using a quantum statistical zero-knowledge hard problem \cite{W2}. Kobayashi proved this equivalence for perfect/statistical/computational zero-knowledge by the direct construction of simulators \cite{Ko}. Multi-prover interactive proof systems with entangled provers and the complexity class $\mathsf{MIP^*}$ were introduced by Cleve et al. \cite{CHTW}. $\mathsf {MIP^*}$ is radically different from classical $\mathsf{MIP}$ due to the non-locality. Quantum multi-prover interactive proofs, where the verifier is also quantum, were introduced by Kobayashi and Matsumoto \cite{KM}, and the corresponding class is denoted by $\sf QMIP^*$.

Despite the importance of quantum zero-knowledge and quantum multi-prover interactive proof systems, there is little research about quantum multi-prover interactive proof systems with zero-knowledge. Recently Chiesa et al. \cite{CFGS} proved that $\mathrm{MIP^*}$ protocols with perfect zero-knowledge can compute $\mathsf{NEXP}$, which is the only literature that analyzed the power of $\mathsf{MIP^*}$  with zero-knowledge as far as the author knows. 

\subsection{Results}
In this paper, we investigate complexity-theoretical properties of $\rm QMIP^*$ protocols with zero-knowledge, different from $\rm MIP^*$ protocols with zero-knowledge in \cite{CFGS}. The first result is the elimination of the honest condition of quantum multi-prover zero-knowledge systems ($\rm QMZK^*$) without any assumptions. 
\begin{theo}[Informal version of Theorem 4]
Any language L computed by honest zero-knowledge quantum multi-prover interactive proof systems can be computed by general zero-knowledge quantum multi-prover interactive proof systems.
\end{theo}
We analyze the public coin protocols of honest verifier quantum multi-prover zero-knowledge systems, which is the analogue of the public coin $\rm{QMIP^*}\ systems$ \cite{KKMV}, since we can convert any honest-verifier $\rm QMZK^*$ protocol into such a restricted protocol. Comparing the systems in \cite{KKMV}, the new difficulty for the multi-prover zero-knowledge case is that the malicious verifier may change messages for each prover even if the honest verifier only sends the same one bit to all provers. The main ingredient in our proof is what we call the GHZ test: it replaces the coin in the public coin protocol of $\sf{QMIP^*}$ \cite{KKMV} by a GHZ state shared by the provers. 
We analyze mainly computational quantum zero-knowledge systems, but this result also holds for statistical quantum zero-knowledge systems. In the main text, we add one prover to eliminate the honest condition, but this addition is not necessary. We prove that the honest condition can be eliminated without increasing the number of the provers combining our GHZ test with the rewinding of \cite{Ko,W2} in Appendix. 

The second result is the construction of computational quantum zero-knowledge systems for $\sf QMIP^*$ with natural computational assumptions.  
\begin{theo}
Let $p$ be any positive integer. If there exists an unconditionally binding and ]computational hiding bit commitment scheme, every language computed by quantum interactive proof systems with $p$ provers has a computational quantum zero-knowledge interactive proof system with $p+1$ provers.
\end{theo}

The main new tool is the Local Hamiltonian based Interactive Protocol (called the LHI protocol in this paper). Local checking of the history of the computation is a very important tool in complexity theory. The most famous example in quantum complexity theory is the Local Hamiltonian problem for $\sf QMA$ \cite{KSV}, which is an analogue of SAT for NP. The LHI protocol is a variation of local checking of the history of the computation for $\rm QMIP^*$ protocols. Recently local checking of the history states of $\rm QMIP^*$ protocols was studied by Ji \cite{Ji1} and Fitzsimons et al.\cite{FJSY}. Their purpose is to construct efficient $\rm MIP^*$ protocols. On the other hand, the verifier in our protocol directly handles quantum history states.

In the LHI protocol,  what the verifier does consists of only sending one message to provers and one measurement on the state received from provers by one local Hamiltonian. Then we apply the technique of quantum zero-knowledge proof for $\sf{QMA}$ by Broadbent, Ji, Song and Watrous \cite{BJSW} to the LHI protocols. In their protocol for the Local Hamiltonian problem, the honest prover has to encode the witness of the Local Hamiltonian problem, to send the commitment of the encoding key, and to convince the verifier that the output of the verifier on the encoded witness and the commitment correspond to yes instances by the zero-knowledge protocol for NP. The process of encoding, committing, and convincing the verifier of the correctness of the witness for the Local Hamiltonian problem can be also applied to encode and check the history state in the multi-prover interactive proof.
The protocol also needs the replacement of coins by GHZ states which is used to prove our first result. 

Chiesa et al. \cite{CFGS} constructed the perfect zero-knowledge $\mathrm{MIP^*}$ protocol for $\mathsf{NEXP}$. The upper bound of $\mathsf{MIP^*(=QMIP^*}$\cite{RUV}), however, is not obvious and $\mathsf{QMIP^*}$ may not be in $\mathsf{NEXP}$. Hence their result does not imply the existence of zero-knowledge protocols for $\mathsf{QMIP^*}$. On the other hand, our protocol assumes the existence of some computational tools and our protocol is only the computational quantum zero-knowledge system. Chiesa et al.'s protocol, however, is perfect zero-knowledge. Hence our result is incomparable to the  result of \cite{CFGS}.

The organization of this paper is as follows. Section 2 is the preliminary section, especially the definition of quantum multi-prover zero-knowledge is here. In Section 3, we prove the elimination of  the honest condition. In Section 4, we construct the zero-knowledge protocol for $\sf QMIP^*$. In Section 5, we discuss some open problems.   In Appendix, we prove the additional result that any honest $\rm QMZK^*$ protocol can be converted into a general $\rm QMZK^*$ protocol without increasing of the number of provers.

\section{Preliminaries}
We assume that the reader is familiar with basic quantum computation \cite{KSV,NS1} and quantum interactive proof systems \cite{VW}. Our main results heavily use results from \cite{BJSW,KKMV}.
\subsection{Notations}
For a string $x$, $|x|$ denotes the length of $x$. For integers $n,l$,  define $[n]=\{1,2,...,n\}$, $[n,l]=\{n,n+1,...,l\}$.  In this paper, $p$ denotes the number of provers, and  $m$ denotes the number of turns in interactive proof systems. Both are the polynomial size of $|x|$. $poly$ denotes some polynomial in $|x|$. $exp$ denotes some exponential of a polynomial in $|x|$. Usually $c$ is a completeness parameter and  $s$ is a soundness parameter in interactive proof systems. A function $\epsilon(|x|)$ is negligible if for any polynomial $g(|x|)$, $\epsilon(|x|)<1/g(|x|)$ holds for any sufficiently large $|x|$.  For a vector $|\psi\rangle$, $\||\psi\rangle\|$ denotes the norm of $|\psi\rangle$. For density matrices $\rho,\sigma$, $\|\rho-\sigma\|$ denotes the trace distance between $\rho$ and $\sigma$. $|+\rangle,|-\rangle$ denote $\frac{1}{\sqrt 2}(|0\rangle+|1\rangle),\frac{1}{\sqrt 2}(|0\rangle-|1\rangle)$. The measurement by a binary-outcome POVM $\{M,\mathrm{Id}-M\}$ is often called as the measurement by $M$. In particular, if $M=|\psi\rangle\langle\psi|$ for a state $|\psi\rangle$, the measurement is called as the projection onto $|\psi\rangle$.
\subsection{Universal gate and Clifford Hamiltonians}
The circuit of the verifier consists of tensor products of two Hadamard gates and controlled phase gates, $\{H\otimes H,\Lambda(P)\}$. The set $\{H,\Lambda(P)\}$ is universal,  and hence $\{H\otimes H,\Lambda(P)\}$ is also universal. This setting is necessary to apply the result of \cite{BJSW}. They used restricted forms of Hamiltonians, called Clifford Hamiltonians. Clifford Hamiltonians are Hamiltonians realized by Clifford operations, followed by a standard basis measurement.
\subsection{$\mathrm{QMIP^*}$ systems}\label{subsectionQMIP*}
A quantum $p$-prover interactive proof system consists of the following data; a verifier with a private register $\mathsf{V}$, $p$ provers with private registers $\mathsf{P}_1,\mathsf{P}_2,...,\mathsf{P}_p$, and message registers $\mathsf{M}_1,\mathsf{M}_2,...,\mathsf{M}_p$. All message registers have $l$ qubits. One of qubits of $\mathsf V$ is the output qubit. At the beginning of the protocol, $(\mathsf{V},\mathsf{M}_1,\mathsf{M}_2,...,\mathsf{M}_p)$ are initialized to $|0\cdots0\rangle$. Provers can set any state $|\psi\rangle$ in $\mathsf{P}_1,\mathsf{P}_2,...,\mathsf{P}_p$.  During the protocol, the verifier and the provers apply their operations onto $(\mathsf{V},\mathsf{M}_1,\mathsf{M}_2,...,\mathsf{M}_p$,$\mathsf{P}_1,\mathsf{P}_2,...,\mathsf{P}_p)$ alternatively. The verifier applies the operations onto $(\mathsf{V},\mathsf{M}_1,\mathsf{M}_2,...,\mathsf{M}_p)$. Prover $p'$ applies the operations onto $(\mathsf{M}_{p'},\mathsf{P}_{p'})$. The verifier's operations can be constructed in polynomial time. Prover's operations have no such restrictions. After applying the last operation of the verifier, the verifier measures the output qubit in the computational basis. If the output is  $|1\rangle$, then the verifier accepts. Denote this probability as $p_{acc}$.
The $m$-turn verifier for the quantum $p$-prover interactive proof system applies circuits $\{V^1,...,V^{\lceil (m+1)/2 \rceil}\}$ which can be constructed in polynomial time. The verifier uses $V^j$ at the $j$-th verifier's turn. The $m$-turn provers apply operations $\{P_i^1,...,P_i^{\lfloor (m+1)/2 \rfloor}\}_{i=1,...,p}$, where $P_i^j$ is the unitary operator done by the $i$-th prover in the provers' $j$-th  turn. Provers' strategy is the tuple of ($|\psi\rangle,\{P_i^1,...,P_i^{\lfloor (m+1)/2 \rfloor}\}_{i=1,...,p}$). $(V,\{P_1,...,P_p\},|\psi\rangle)$ denotes the corresponding interactive proof system, that is, the tuple of the set of circuits of the verifier and the provers' strategy.
\begin{defi}A language $L$ is in $\mathsf{QMIP^*}(p,m,c,s)$ if  there is an $m$-turn verifier V, for any $x$,
\begin{enumerate}
\item If $x\in L$, then there are $m$-turn provers $P_1,...,P_p$ and a state $|\psi\rangle$  (called honest provers), $p_{acc}\ge c$.
\item If $x\notin L$, then, for any $m$-turn provers $P_1',...,P_p'$ and any state $|\psi'\rangle$, $p_{acc}\le s$.
\end{enumerate}
\end{defi}

\subsection{Quantum zero-knowledge}
First, we introduce the indistinguishability of states and channels. We note that the definitions of indistinguishability allow the auxiliary quantum state $\sigma$. 
\begin{defi}[Indistinguishable states]
Let $S\subseteq\{0,1\}^*$ be a language, and let $r$ be a function bounded by a polynomial. For each $x\in S$, $\rho_x,\zeta_x$ are $r(|x|)$ qubits states. The ensembles $\{\rho_x|x\in S\}$ and $\{\zeta_x|x\in S\}$ are computationally indistinguishable if for any function $s,k$ bounded by a polynomial, for any circuit $Q$, the size of which are $s(|x|)$ and the output of which is $0$ or $1$, and for any $k(|x|)$ qubits state $\sigma$, there is a negligible function $\epsilon$, such that
\begin{equation*}
\|\mathrm{Pr}[Q(\rho_x\otimes\sigma)=1]-\mathrm{Pr}[Q(\zeta_x\otimes\sigma)=1]\|\le\epsilon(|x|).
\end{equation*}
\end{defi}
\begin{defi}[Indistinguishable channels]
Let $S\subseteq\{0,1\}^*$ be a language, and let $q,r$ be functions bounded by a polynomial. For each $x\in S$, $\Psi_x,\Phi_x$ are channels from $q(|x|)$ qubits to $r(|x|)$ qubits for each $x\in S$. The ensembles $\{\Psi_x|x\in S\}$ and $\{\Phi_x|x\in S\}$ are indistinguishable if for any functions $s,k$ bounded by a polynomial, for any circuit $Q$ the size of which is $s(|x|)$, acts on $r(|x|)+k(|x|)$ qubits and outputs $0$ or $1$, and for any $r(|x|)+k(|x|)$ qubits state $\sigma$, there is a negligible function $\epsilon$, such that
\begin{equation*}
\|\mathrm{Pr}[Q(\Psi_x\otimes \mathrm{Id}) (\sigma)=1]-\mathrm{Pr}[Q(\Phi_x\otimes \mathrm{Id})(\sigma)=1]\|\le\epsilon(|x|).
\end{equation*}
\end{defi}
Next, we introduce quantum zero-knowledge proof systems. Let $(V,\{P_1,...,P_p\},|\psi\rangle)$ be a quantum multi-prover interactive proof system. Let $view_{(V,\{P_1,...,P_p\},|\psi\rangle)}(x,j)$ be the state in register $(\mathsf{V},\mathsf{M}_1,...,\mathsf{M}_p)$ at the $j$-th turn of the interactive proof. A malicious verifier $V'$ is a circuit family constructed in polynomial time, which interacts with $\{P_1,...,P_p\}$ instead of $V$. $V'$ has an additional input register $\mathsf W$ and an output register $\mathsf Z$. The interaction of $(V',\{P_1,...,P_p\},|\psi\rangle)$ until the $j$-th turn can be interpreted as a channel $\langle V',\{P_1,...,P_p\},|\psi\rangle \rangle(x,j)$ from $\mathsf W$ to $\mathsf Z$ decided by $(V',x,j)$.
\begin{defi}[Quantum Multi-Prover Zero-Knowledge Systems]
\sloppy A quantum interactive proof $(V,\{P_1,...,P_p\},|\psi\rangle)$ for a language $L$ is computationally zero-knowledge if for any verifier $V'$, there is an ensemble of channels $\{S_{V'}(x,j)\}$ (called a simulator) constructed in polynomial time, and $\{\langle V',\{P_1,...,P_p\},|\psi\rangle \rangle(x,j)|x\in L\}$ and $\{S_{V'}(x,j)|x\in L\}$ are indistinguishable.

Let $\sf {QMZK^*}$$(p,m,c,s)$ be the class decided by quantum multi-prover zero-knowledge systems with $p$ provers, $m$ turns, completeness $c$ and soundness $s$. We define $\mathsf{QMZK^*}:=\cup_{p,m:poly,c-s\ge\frac{1}{poly}}\mathsf{QMZK^*}(p,m,c,s)$.
\end{defi}

\begin{defi}[Honest Verifier Quantum Multi-Prover Zero-Knowledge Systems]
A quantum $p$-prover interactive proof $(V,\{P_1,...,P_p\},|\psi\rangle)$ for a language $L$ is honest verifier zero-knowledge if for the honest verifier $V$, there is an ensemble of states $\{\Psi_{V}(x,j)\}$ constructed in polynomial time and  $\{view_{(V,\{P_1,...,P_p\},|\psi\rangle)}(x,j)|x\in L\}$ and $\{\Psi_{V}(x,j)|x\in L\}$ are indistinguishable.

\sloppy Let $\sf {HVQMZK^*}$$(p,m,c,s)$ be the class decided by honest verifier quantum multi-prover zero-knowledge systems with $p$ provers, $m$ turns, completeness $c$ and soundness $s$. We define $\mathsf{HVQMZK^*}:=\cup_{p,m:poly,c-s\ge\frac{1}{poly}}\mathsf{HVQMZK^*}(p,m,c,s)$.
\end{defi}

We note a lemma used to eliminate the honest condition. 
\begin{lemm}[\cite{Aa}]
\label{2PO}
Assume a POVM on a state $\rho$ with two outcomes has the outcome 1 with probability $\epsilon$. Denote the state after it outputs 0 by $\rho_0$. Then $\|\rho-\rho_0\|\le\sqrt\epsilon$.
\end{lemm}

\subsection{Cryptographic tools}
Here we introduce a few cryptographic tools used in the $\mathrm{QMZK^*}$ protocol to apply the technique in \cite{BJSW}. Assuming the existence of commitment schemes,   computational quantum zero-knowledge proofs for NP \cite{W2} and quantumly secure coin-flipping \cite{DL} exist.
\begin{defi}[commitment schemes]
A quantum computationally commitment scheme for an alphabet $\Gamma$ is an ensemble of functions $\{f_n: \Gamma\times \{0,1\}^{p(n)}\rightarrow\{0,1\}^{q(n)},\ n\in \mathbb{N}\}$ computable in polynomial time such that following holds;
\begin{itemize}
\item Unconditionally binding property. For all $n\in\mathbb{N},a,b\in\Gamma$ and $r,s\in\{0,1\}^{p(n)}$, if $f_n(a,r)=f_n(b,s)$, then $a=b$.
\item Quantum computationally concealing property. For all $a\in\Gamma$ and $n\in\mathbb{N}$, define
\begin{equation}
\rho_{a,n}=\frac{1}{2^{p(n)}}\sum_{r\in\{0,1\}^{p(n)}}|f_n(a,r)\rangle\langle f_n(a,r)|.
\end{equation}
For all  $a,b\in\Gamma$, the ensemble $\{\rho_{a,n}|n\in \mathbb{N}\}$ and $\{\rho_{b,n}|n\in \mathbb{N} \}$ are computationally indistinguishable.
\end{itemize}
\end{defi}
\subsubsection{Coin-Flipping}
A coin-flipping protocol is an interactive process that allows two parties to jointly toss random
coins. We only make use of one specific coin-flipping protocol, namely Blum's coin-flipping protocol \cite{Bl} in which an honest prover commits to a random $y\in\{0,1\}$, the honest verifier selects $z\in\{0,1\}$ at random, the prover reveals $y$, and the two participants agree that the random bit generated is $r=y\oplus z$. Damg\r{a}rd and Lunemann \cite{DL} proved that Blum's coin-flipping protocol is quantum-secure, assuming a quantum-secure commitment scheme.  

\section{Elimination of the honest condition}
In this section, we eliminate the honest condition of $\rm QMZK^*$ protocols. This needs no computational assumptions. First, we reduce the number of turns. The reduction is almost the same as Kempe et al.  \cite{KKMV}. There are no new contents except for the addition of honest zero-knowledge, while it is shown obviously. Hence we only give the statement. 
\begin{theo} Let $\sqrt s\le c$. Then, $\sf {HVQMZK^*}$$(p,4m+1,c,s)\subseteq$ $\sf{HVQMZK^*}$$(p,2m+1,\frac{1+c}{2},\frac{1+\sqrt s}{2})$.
\end{theo}
The next proposition is easily derived by applying Theorem 3 to a protocol obtained from repeating the original protocols with threshold value decisions.
\begin{prop}
$\sf {HVQMZK^*}$$(p,m,c,s)\subseteq$ $\sf{HVQMZK^*}$$(p,3,1-\frac{1}{exp},1-\frac{1}{poly})$.
\end{prop}
\begin{proof}[Proof(Sketch)] $\sf {HVQMZK^*}$$(p,m,c,s)\subseteq$ $\sf{HVQMZK^*}$$(p, poly\cdot m,1-\frac{1}{exp},\frac{1}{exp})$ is easily shown by repeating the original $\sf {HVQMZK^*}$$(p,m,c,s)$ protocol $poly$ times, threshold value decisions and applying Chernoff bounds. Apply Theorem 3 to $\sf{HVQMZK^*}$$(p, poly\cdot m,1-\frac{1}{exp},\frac{1}{exp})$.
\end{proof}
The next theorem is our first result. The reason why we add one more prover (prover 0) is to prevent prover $1,\ldots,p$ from acting on register $\mathsf{V}$ (the verifier's private register) by making prover 0 have the register $\sf V$. 

\begin{theo}
$\sf{HVQMZK^*}$$(p,3,1-\epsilon,1-\delta)\subseteq \sf{QMZK^*}$$(p+1,2,1-\frac{\epsilon}{2},1-\frac{\delta^2}{8})$.
\end{theo}

\begin{figure}
\hrulefill\\
Let $V_1,V_2$ be the circuit that the verifier of a three-turn protocol for $\mathsf{HVQMZK^*}(p,3,1-\epsilon, 1-\delta)$ applies after receiving message registers $(\mathsf{M}_1,\mathsf{M}_2,...,\mathsf{M}_p)$ from provers $1,...,p$, respectively, at the first turn and the third turn. The two-turn protocol is the following;
\begin{enumerate}
\item Select $b\in\{0,1\}$ uniformly at random. Send $b$ to provers $1,...,p$, and nothing to prover $0$.
\item Do the following tests depending on $b=0,1$:
\begin{enumerate}
\item 
($b=0$: Forward test) Receive $\mathsf {V}$ from prover 0 and receive $\mathsf {M}_i$ from prover $i$ for $i=1,...,p$. Apply $V_2$ to $(\mathsf{V},\mathsf{M}_1,\mathsf{M}_2,...,\mathsf{M}_p)$. Accept if the original verifier of the three-turn protocol accepts the state in $(\mathsf{V},\mathsf{M}_1,\mathsf{M}_2,...,\mathsf{M}_p)$. Reject otherwise.\\
\item 
($b=1$: Backward test) Receive $\mathsf {V}$ from prover 0 and receive $\mathsf {M}_i$ from prover $i$ for $i=1,...,p$. Apply $(V_1)^\dagger$ to $(\mathsf{V},\mathsf{M}_1,\mathsf{M}_2,...,\mathsf{M}_p)$. Accept if all qubits in $\mathsf{V}$ are $|0\rangle$, otherwise reject.
\end{enumerate}
\end{enumerate}
\hrulefill
\caption{$(p+1)$-prover two-turn honest verifier zero-knowledge protocol constructed from a three-turn protocol for $\mathsf{HVQMZK^*}(p,3,1-\epsilon, 1-\delta)$.} 
\label{2turn}
\quad\\
\hrulefill\\
$\mathsf{V},\mathsf{M}_i$ are the same as Figure \ref{2turn}. In addition, we use $\mathsf{G}_0, \mathsf{G}_1, \mathsf{G}_2,...,\mathsf{G}_p$, all of which are 1 qubit registers.
At the start, prover $i$ has $\mathsf{G}_i$.
\begin{enumerate}
\item Select $b'\in\{0,1\}$ uniformly at random. Send $b'$  to provers $1,...,p$, and nothing to prover 0.
\item Do the following tests depending on $b'=0,1$:
\begin{enumerate}
\item
($b'=0$:GHZ test)\\
Prover 0 sends $\mathsf{V,G_0}$. Prover $i$ sends $\mathsf{G}_i$, for $i=1,...,p$. Measure ($\mathsf{G}_0, \mathsf{G}_1, \mathsf{G}_2,...,\mathsf{G}_p$) by  the projection onto $\frac{1}{\sqrt2} (|0^{p+1}\rangle+|1^{p+1}\rangle)$. If $\frac{1}{\sqrt2} (|0^{p+1}\rangle+|1^{p+1}\rangle)$ is measured, then accept.
\item
($b'=1$:History test)\\
Prover 0 sends $\mathsf{V,G_0}$. Prover $i$ sends $\mathsf {M}_i$ and $\mathsf{G}_i$, for $i=1,...,p$. Verifier measures  $\mathsf{G}_0$ in the computational basis.  Let $b$ be the output of this measurement. If $b=0$, the verifier measures ($\mathsf{V},\mathsf{M}_1,...,\mathsf{M}_p)$ as the verifier in Figure \ref{2turn} does at the Forward test.  If $b=1$, the verifier measures ($\mathsf{V},\mathsf{M}_1,...,\mathsf{M}_p)$ as the verifier in Figure \ref{2turn} does at the Backward test .
\end{enumerate}
\end{enumerate}
\hrulefill
\caption{General zero-knowledge protocol based on the protocol in Figure \ref{2turn}.}
\label{2gz}
\end{figure}

\begin{proof}
First, we construct a $(p+1)$-prover two-turn protocol from a three-turn protocol for $\sf{HVQMZK^*}$$(p,3,1-\epsilon,1-\delta)$, in which message to provers is the same classical one bit except for prover 0, who receives no messages. The two-turn protocol is described in Figure \ref{2turn}. The main protocol of $\sf{QMZK^*}$$(p+1,2,1-\frac{\epsilon}{2},1-\frac{\delta^2}{8})$ is described in Figure \ref{2gz}. The proof of the correctness of the two-turn protocol in Figure \ref{2turn} is the same as the proof of Lemma 5.4 in \cite{KKMV}. We only discuss the conversion of a two-turn protocol in Figure \ref{2turn} into a general zero-knowledge protocol in Figure \ref{2gz}.\\\\
$\bf Completeness$: The conversion of the two-turn protocol to the general zero-knowledge preserves completeness since the provers only have to do the  GHZ test and the history test honestly.\\\\
$\bf Soundness$ :  We construct a $\mathrm{HVQMZK^*}$ protocol from a $\mathrm{QMZK^*}$ protocol  with a sufficient acceptance probability. Let $U_g^i,U_h^i$ be the unitary operators that prover $i$, for $i=1,\ldots, p$, does for the  GHZ test and the history test, respectively. (Prover 0 receives no messages, and hence we can assume prover 0 does only the identity operator.)

Let $\rho'':=\rm{Tr}_{GHZ}$$(\Pi_{GHZ}(\otimes_{i=1}^p U_g^i)\rho(\otimes_{i=1}^p (U_g^i)^\dagger)\Pi_{GHZ})$. Here, $\rho$ is the initial state shared by the provers in the $\mathrm{QMZK^*}$ protocol and $\Pi_{GHZ}$ is a projection onto $\frac{|0^{p+1}\rangle+|1^{p+1}\rangle}{\sqrt2}$ of registers ($\mathsf{G}_0, \mathsf{G}_1, \mathsf{G}_2,...,\mathsf{G}_p$) which is used in the  GHZ test, and $\rm Tr_{GHZ}$ is the partial trace operation on registers ($\mathsf{G}_0, \mathsf{G}_1, \mathsf{G}_2,...,\mathsf{G}_p$). Let $\rm Tr$$\rho''=1-\epsilon_1$, where $\epsilon_1$ is the rejection probability of the  GHZ test. Let $\frac{\rho''}{1-\epsilon_1}:=\rho'$. We use this state as the initial state of provers. Let $\epsilon_2$ be the rejection probability of the history test.

Now we construct the strategy of  the two-turn $\rm{HVQMZK^*}$ protocol from the strategy of $\rm{QMZK^*}$ protocol; The private registers of prover $i$ are $\mathsf{P}_i,\mathsf{G}_i$, for $i=1,...,p$. Initially prover 0 has only register $\mathsf{V}$. The initial state in ($\mathsf{P}_1,...,\mathsf{P}_p,\mathsf{M}_1,...,\mathsf{M}_p$) is $\rho'$. If the verifier sends $b$ to prover $i$ ($i=1,...,p$), prover $i$ sets $|b\rangle$ in register $\mathsf{G}_i$, applies $U^i_h{U^i_g}^\dagger$ and sends $\mathsf{M}_i$ to the verifier. Prover 0 always sends $\mathsf{V}$.

The rejection probability $p_{rej,hv}$ of this $\mathrm{HVMQZK^*}$ protocol is bounded by the rejection probability $p_{rej}$ of the general zero-knowledge protocol as follows. Here, $\Pi_{rej}$ is the projection onto rejected states.

\begin{equation}
\begin{split}
p_{rej,hv}&:=\frac{1}{2}\sum_{b\in\{0,1\}}\mathrm{Tr}[\Pi_{rej}(\otimes_{i=1}^p U_h^i {U_g^i}^\dagger) (|b^p\rangle\langle b^p|\otimes \rho') (\otimes_{i=1}^p U_g^i{U_h^i}^\dagger)]\\
&=\frac{1}{1-\epsilon_1}\mathrm{Tr} [\Pi_{rej} (\otimes_{i=1}^p U_h^i{U_g^i}^\dagger)\Pi_{GHZ}(\otimes_{i=1}^p U_g^i)\rho (\otimes_{i=1}^p {U_g^i}^\dagger)\Pi_{GHZ}(\otimes_{i=1}^p U_g^i{U_h^i}^\dagger)]\\
&\le \sqrt\epsilon_1+\mathrm{Tr} [\Pi_{rej}(\otimes_{i=1}^p U_h^i {U_g^i}^\dagger)(\otimes_{i=1}^p U_g^i)\rho (\otimes_{i=1}^p {U_g^i}^\dagger)(\otimes_{i=1}^p U_g^i{U_h^i}^\dagger)]\\
&=\sqrt\epsilon_1+\mathrm{Tr}[\Pi_{rej}(\otimes_{i=1}^p U_h^i)\rho(\otimes_{i=1}^p {U_h^i}^\dagger)]\\
&=\sqrt\epsilon_1+\epsilon_2\le\sqrt\epsilon_1+\sqrt\epsilon_2\le2\sqrt {2p_{rej}}.
\end{split}
\end{equation}
The first equality follows from the definition of $\rho'=\frac{1}{1-\epsilon_1}\rho''$$=\frac{1}{1-\epsilon_1}\rm{Tr}_{GHZ}$$(\Pi_{GHZ}(\otimes_{i=1}^p U_g^i)\rho(\otimes_{i=1}^p (U_g^i)^\dagger)\Pi_{GHZ})$ and that $\frac{1}{2}(|0^p\rangle\langle0^p|+|1^p\rangle\langle1^p|)$ equals to the $p$-qubits substate 
of the $(p+1)$-qubit GHZ state. The first inequality follows since $\|\rho'-\rho\|\le\sqrt\epsilon_1$ holds by regarding $\{\Pi_{GHZ},\mathrm{Id}-\Pi_{GHZ}\}$ as a two-outcome POVM and applying Lemma 1.\\\\
$\bf Zero$ $\bf knowledge$: Before the verifier sends messages, the verifier gets no qubits from provers. Hence it is sufficient to construct a simulator of the malicious verifier for each of the verifier's possible messages. 

First, we observe that it is sufficient to prove only the case that the malicious verifier sends the same $b'$ to all provers except for prover 0. Here, we define the action of the honest prover $j$ in Figure 2 as follows;
 \begin{enumerate}
\item If the prover $j$ receives $b'=1$, then he/she measures the register $\mathsf{G}_j$ in the computational basis. Let the outcome be $b$. He/she acts as the prover $j$ in Figure 1 and sends $\mathsf{M}_j,\mathsf{G}_j$.
\item If the prover $j$ receives $b'=0$, then he/she sends $\mathsf{G}_j$. 
\end{enumerate}
This definition preserves the completeness parameter. Assume that the verifier sends $b'=1$ (history test) to prover $j$ and $b'=0$ (GHZ test) to prover $i$. Prover $i$ sends register $\mathsf{G}_i$ to the verifier, but the qubit in $\mathsf{G}_i$ is the same $|0\rangle$ or $|1\rangle$ of $\mathsf{G}_j$. Hence, if the case that the verifier sends $b'=1$ to all provers can be simulated, the case that prover $i$ receives $b'=0$ can be also simulated since it is sufficient to simulate any verifier who sends $b'=1$ to all provers and aborts $\mathsf{M}_i$. 

The case that the verifier sends $b'=0$ to provers $1,...,p$ can be also simulated since it is sufficient to simulate the verifier who sends $b'=1$ to all provers, aborts $\mathsf{M}_i,\mathsf{G}_i$, and prepares $\frac{1}{\sqrt2} (|0^{p+1}\rangle+|1^{p+1}\rangle)$. The case the verifier sends $b'=1$ to all provers is obviously zero-knowledge since the verifier receives only the copies of a uniform random 1 bit and the state that the honest verifier in Figure \ref{2turn} receives.
\end{proof}
Theorem 4 needs an additional prover to eliminate the honest condition, but by the rewinding technique of \cite{Ko,W2}, we can eliminate the honest condition without additional provers. This result essentially follows by combining the proof of Theorem 4 and the rewinding, and our main contribution is described in the proof of Theorem 4. Hence we prove this result in Appendix. To prove this result, we have to restrict the verifier's message to one public classical bit by the technique of \cite{KKMV}. Let $\sf{HVQMZK^*}_{pub,1}$$(k,3,c,s)$ be the class of problems verified by such $\rm HVQMZK$ systems. We omit the proof of the next theorem since it is also the same as \cite{KKMV}.
\begin{theo}Let $\sqrt s\le c$. Then, $\sf {HVQMZK^*}$$(p,3,c,s)\subseteq$ $\sf{HVQMZK^*}_{pub,1}$$(p,3,\frac{1+c}{2},\frac{1+\sqrt s}{2})$.
\end{theo}
We prove the next theorem in Appendix. This result also needs the GHZ test, and rewinding alone is not sufficient to prove this result.
\begin{theo}
$\sf{HVQMZK^*}_{pub,1}$$(p,3,1-\epsilon,1-\delta)\subseteq \sf{QMZK^*}$$(p,3,1-\frac{\epsilon}{2},1-\frac{\delta^2}{8})$.
\end{theo}
\section{LHI protocols and computational quantum zero-knowledge systems}
In this section, we construct the zero-knowledge system for $\sf QMIP^*$. The construction of a zero-knowledge system for $\sf QMIP^*$ consists of the following three steps;
\begin{enumerate}
\item Construction of the Local Hamiltonian based Interactive protocol (LHI protocol) corresponding to the $p$-prover $\mathrm{QMIP^*}$ protocol, which extends the protocol for checking the local Hamiltonian problem corresponding to a QMA protocol in \cite{BJSW} to the $\rm QMIP^*$ case.
\item Construction of what we call the LHI+ protocol, which replaces the uniform random queries  by the GHZ test.
\item Zero-knowledge protocol for $\sf QMIP^*$ based on the LHI+ protocol.
\end{enumerate}
Our main contribution is step 1 and 2. Step 3 is a direct application of the technique of Broadbent et al. \cite{BJSW}. The analysis of step 3 is almost the same as \cite{BJSW}, and hence we only point out the main differences.
 In Section 4.1, we construct the LHI protocol and prove its validity. We only consider three-turn $\mathrm{QMIP^*}$ protocols, as this does not lose generality due to Lemma 4.2 of \cite{KKMV}.  In Section 4.2, we provide the LHI+ protocol and prove its validity. In Section 4.3, we overview the quantum zero-knowledge protocol for $\sf{QMA}$ and explain why it works for the LHI+ protocol. In Section 4.4, we construct the final zero-knowledge protocol. We analyze the final protocol in Section 4.5.

\subsection{Construction of the LHI protocol\label{LH2;sec}}
First, we construct a local Hamiltonian based Interactive protocol (LHI protocol) for a three-turn $\mathrm{QMIP^*}$  protocol $\cal{P}$. Intuitively, the LHI protocol checks the history state of the calculation of the interactive proof system in Figure \ref{1com}, which is transformed from the original three-turn protocol $\cal{P}$ with completeness $1-\frac{1}{exp}$ and soundness $1-\frac{1}{poly}$ in Figure \ref{3turn}. We can assume such completeness/soundness errors on $\cal{P}$ due to Lemma 4.2 of \cite{KKMV}. This transformation is done to locally check the communication of the interactions. Let $V_0=U_{t_0}\cdots U_1$ and  $V_1=U_T\cdots U_{t_0+2pl+p+1}$ be  the circuits which the verifier in $\cal{P}$ uses (see Figure \ref{3turn}).  Here, $l$ is the length of the message register $\mathsf{M}_i$  between the $i$-th prover and the verifier in $\cal{P}$, $t_0$ is the number of gates in $V_0$ and $(T-(t_0+2pl+p+1)+1)$ is the number of gates in $V_1$. The reason why the indices of $V_0,V_1$ are not successive is the increase of communication turns by the transformation from Figure \ref{3turn} to Figure \ref{1com}.  This transformation is necessary to make prover 0 send message registers $\mathsf{M}_1,...,\mathsf{M}_p$ in the LHI protocol to prevent the malicious attack on message registers by provers $1,...,p$ depending on the verifier's message.
We can assume all $\mathsf{M}_i$ have the same length. The LHI protocol uses $T$ registers $\mathsf{C}_1,...,\mathsf{C}_T$ and $p$ registers $\mathsf{Me}_1,...,\mathsf{Me}_p$ where each of those $(T+p)$ registers consists of a single qubit, in addition to  $\mathsf{M}_1,...,\mathsf{M}_p$, and $\mathsf{V}$. Intuitively, $\mathsf{C}_1,...,\mathsf{C}_T$ are the time counters, $\mathsf{Me}_1,...,\mathsf{Me}_p$ are 1 qubit message registers. The LHI protocol is described in Figure \ref{LHp}. Without loss of generality, we can assume $T^5\ll\epsilon^{-1}$ by Lemma 4.2 of \cite{KKMV}.

Though intuitively the LHI protocol corresponds to the protocol in Figure \ref{1com}, we show that if the LHI protocol can be accepted with high probability, then the protocol $\cal{P}$ in Figure \ref{3turn} can be accepted with high probability, since we put Figure \ref{1com} for ease of intuitive understanding LHI protocols, but the correctness of Figure \ref{1com} is logically unnecessary and proving the correctness is only redundant. The next lemma states that if there exist provers who pass the LHI protocol, then there exist provers who pass the original three-turn protocol $\cal P$. 

\begin{figure}
\hrulefill\\
\includegraphics[width=0.8\textwidth]{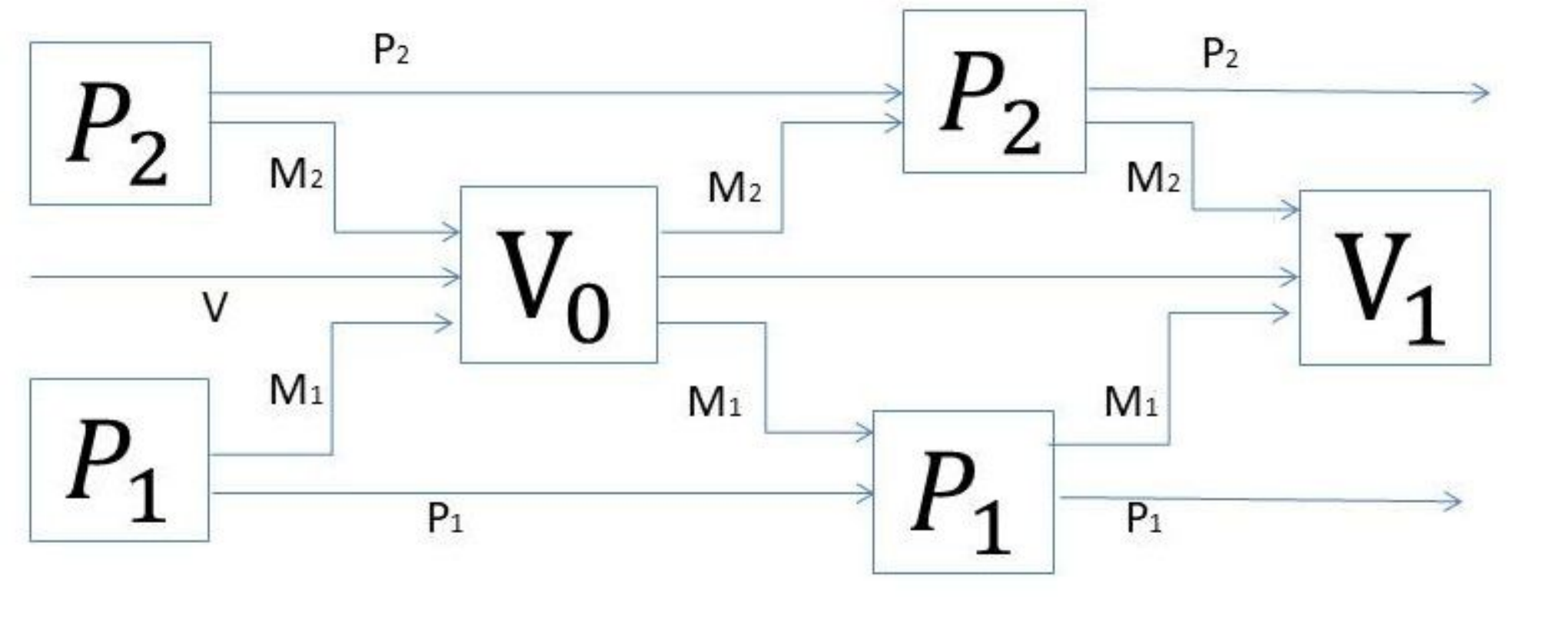}
\caption{Original three-turn protocol $\cal{P}$. In this figure, the number of provers is two.}
\label{3turn}
\hrulefill\\
\end{figure}

\begin{figure}
\hrulefill\\
\includegraphics[width=0.8\textwidth]{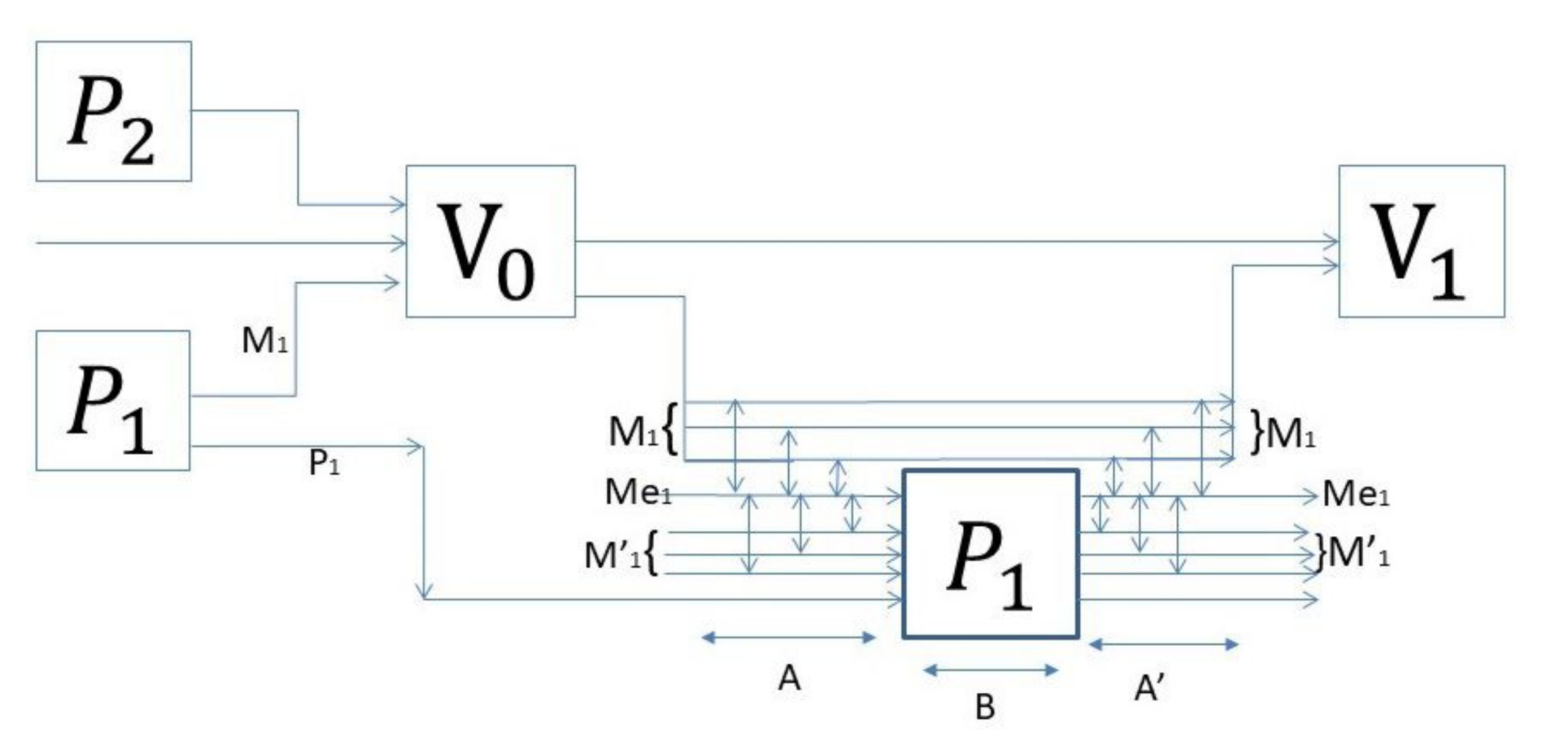}
\caption{Restricting communications to one qubit. In this figure, the length of message register is three. The arrows $\updownarrow$ are SWAP gates. $\mathsf{M'_1}$ is used by honest prover 1 to receive the qubits from the verifier. Note that $\sf{M'_1}$ is introduced to explain the behavior of the honest prover and hence it does not appear in the soundness analysis in the main text. Communications between the verifier and prover 2 are similarly transformed. A,B,A' denote the correspondence to the operators $A,B,A'$ in Eq.(\ref{opr}).}
\label{1com}
\hrulefill
\end{figure}

\begin{lemm}Suppose $T^5\ll\epsilon^{-1}$. If there are provers who pass the protocol in Figure \ref{LHp} with probability $1-\epsilon$, there are provers who pass the original three-turn protocol $\cal{P}$ with probability $1-O(\sqrt[4]{T^5\epsilon})$.
\end{lemm}

\begin{figure}
\hrulefill\\
(A). LHI protocol
\begin{enumerate}
\item The verifier sends $t'\in [2T+n]$ to provers $1,...,p$ and nothing to prover 0.
\item Prover $i$ sends $\mathsf{Me}_i$. Prover 0 sends $\mathsf{V},\mathsf{M}_1,..., \mathsf{M}_p,\mathsf{C}_1,...,\mathsf{C}_T$.
\item The verifier measures these qubits by $H_{t'}$.
\end{enumerate}
Here, $n$ is the number of qubits of $\mathsf{V}$. We can assume $n\le T$.\\
\hrulefill\\
(B). Hamiltonians used in the LHI protocol\\\\
$SWAP _{(k,j),\mathsf{Me}_k}$ is the SWAP operator on the $j$-th qubit of $\mathsf {M}_k$ and $\mathsf{Me}_k$, $CNOT_{\mathsf{C}_{t_0+pl+k},\mathsf{Me}_k}$ is a CNOT operator which is controlled by $\mathsf{C}_{t_0+pl+k}$ and acts on $\mathsf{Me}_k$. $\Pi_{acc}$ is a projection that the verifier in $\cal{P}$ accepts. ${\rm Id}$ is an identity.
The set $\{H_t|t\in [2T+n]\}$ has the following Hamiltonians as described in (a), (b), (b'), (c), (d), (e), and (f).\\\\
In (a), (b), (b'), we define $H_t$ by Eq. (\ref{lh}). 
\begin{equation}
\label{lh}
H_t:=|10\rangle\langle10|_{\mathsf{C}_{t-1},\mathsf{C}_{t+1}}\otimes\frac{1}{2}(-|1\rangle\langle0|_{\mathsf{C}_t}\otimes U'_t- |0\rangle\langle1|_{\mathsf{C}_t}\otimes {U'_t}^\dagger+|0\rangle\langle0|_{\mathsf{C}_t}\otimes \mathrm{Id} +|1\rangle\langle1|_{\mathsf{C}_t}\otimes \mathrm{Id}).
\end{equation}
If $t=1$, we replace $|10\rangle\langle10|_{\mathsf{C}_{t-1},\mathsf{C}_{t+1}}$ by $|0\rangle\langle0|_{\mathsf{C}_{t+1}}$ exceptionally since register $\mathsf{C}_0$ does not exist. Similarly we ignore register $\mathsf{C}_0$ in (e).

The unitary operator $U'_t$ in Eq.(\ref{lh}) for (a), (b), (b') is defined as follows;
\begin{itemize}
\item[(a)](verifier's gate step) $t\in[1,t_0]\cup[t_0+2pl+p+1,T]$: $U'_t=U_t$.
\item[(b)](communication step) $t\in[t_0+1,t_0+pl]$ : Let $t=t_0+(k-1)l+j$ ($1\le k\le p,1\le j\le l$). Then $U'_t:=SWAP _{(k,j),\mathsf{Me}_k}$.
\item[(b')](communication step) $t\in[t_0+pl+p+1,t_0+2pl+p]$: Let $t=t_0+pl+p+(k-1)l+j$. Then $U'_t:=SWAP _{(k,j),\mathsf{Me}_k}$.
\end{itemize}
In (c), (d), (e), (f), we define $H_t$ as follows;
\begin{itemize}
\item[(c)](provers' step) $t=t_0+pl+k\ (k=1,...,p)$:\\
\quad$H_{t_0+pl+k}:= |10\rangle\langle10|_{\mathsf{C}_{t_0+pl+k-1},\mathsf{C}_{t_0+pl+k+1}}\otimes|-\rangle\langle-|_{\mathsf{C}_{t_0+pl+k}} CNOT_{\mathsf{C}_{t_0+pl+k},\mathsf{Me}_k}$.
\item[(d)](measurement by the verifier) $t=T+1$: $H_{T+2pl+p+1}:=|1\rangle\langle1|_{\mathsf{C}_T}\otimes({\rm Id}-\Pi_{acc})$.
\item[(e)](consistency of time counters) $t\in[T+2,2T]$: $H_{t}:=|01\rangle\langle01|_{\mathsf{C}_{t-T-1},\mathsf{C}_{t-T}}$. 
\item[(f)](initialization of $\mathsf{V}$) $t\in [2T+1,2T+n]$: $H_t:=|0\rangle\langle0|_{\mathsf{C}_1}\otimes(\mathrm{Id}-\Pi_0^{t-2T})$. Here, $\Pi_0^{t-2T}$ is a projection onto $|0\rangle$ of the ($t-2T$)th qubit of $\mathsf{V}$.
\end{itemize}
\hrulefill
\caption{Local Hamiltonian based Interactive Protocol for $\cal{P}$ (in fact, for the $\mathrm{QMIP^*}$ protocol in Figure \ref{1com}). (A) is the protocol. (B) shows the Hamiltonians used in the protocol given in (A). }
\label{LHp}
\end{figure}
\begin{proof}
 We construct the $\mathrm{QMIP^*}$ protocol for $\cal{P}$ to compute the original three-turn protocol with acceptance probability at least $1-O(\sqrt[4]{T^5\epsilon})$  based on the protocol in Figure \ref{LHp} with acceptance probability at least $1-\epsilon$. Intuitively, (a),...,(f) in Figure \ref{LHp} check the following; 
 \begin{description}
 \item[(a)]the verifier's gate step
 \item[(b) and (b')]communication step
 \item[(c)]the provers' step
 \item[(d)]the measurement by the verifier
 \item[(e)]the consistency of the time counters
 \item[(f)]the initialization of $\mathsf V$
\end{description}
  Let $|\varphi\rangle$ be the initial state of all the registers in the LHI protocol. Let ${D_{k,j}}$ ($D_{k,j}'$, respectively)  be the operation that prover $k$ does in test (b) with query $t=t_0+(k-1)l+j$ ((b')  with $t=t_0+pl+p+(k-1)l+j$, respectively).  Let $E_k$ be the operation that prover $k$ does in test (c). Denote the Pauli $X$ operator on $\mathsf{Me}_k$ by $X_{\mathsf{Me}_k}$. Denote the projection onto $|0^n\rangle$ of register $\mathsf{V}$ by $\Pi_{0^n}$. 
  
Before we construct the provers' strategy for $\cal{P}$, we define some unitary operators. \sloppy Let $|t\rangle=|1\rangle_{\mathsf{C_1}}|1\rangle_{\mathsf{C_2}}\cdots|1\rangle_{\mathsf{C}_t}|0\rangle_{\mathsf{C_{t+1}}}\cdots|0\rangle_{\mathsf{C}_T}$ be a state in counter registers that corresponds to time $t$. Let
\begin{equation}
\begin{split}
&U_{t_0+(k-1)l+j}:=D_{k,j}^\dagger SWAP _{(k,j),\mathsf{Me}_k}D_{k,j}\ (k\in[1,p],j\in[1,l]),\\
&U_{t_0+pl+p+(k-1)l+j}:={D'_{k,j}}^\dagger SWAP _{(k,j),\mathsf{Me}_k}D'_{k,j}\ (k\in[1,p],j\in[1,l]),\\
&U_{t_0+pl+k}:= {E_k}^\dagger X_{\mathsf{Me}_k}E_k\ (k\in[1,p]),
\end{split}
\end{equation}
and
\begin{equation}
\label{opr}
\begin{split}
&W:=\Sigma_{t\in [0,T]} |t\rangle\langle t| \otimes U_1^\dagger\cdots U_t^\dagger+\Pi_{bad time}\otimes \rm{Id},\\
&A_k:=U_{t_0+(k-1)l+l}\cdots U_{t_0+(k-1)l+1},\\
&B_k:={U}_{t_0+pl+k},\\
&A'_k:={U_{t_0+pl+p+(k-1)l+1}}\cdots {U_{t_0+pl+p+(k-1)l+l}}.
\end{split}
\end{equation}
where $|t\rangle\langle t| \otimes U_1^\dagger\cdots U_t^\dagger$ for $t=0$ is $|0\rangle\langle0|\otimes\rm Id$, and $\Pi_{bad time}$ is a projection onto states rejected by test (e). Intuitively, $A_k,B_k,A'_k$ correspond to A,B,A' in Figure \ref{1com}. Note that $U_t$ for $t\notin [t_0+1,t_0+2pl+p]$ are already defined by the verifier's gates; $V_0=U_{t_0}\cdots U_1$ and $V_1=U_T\cdots U_{t_0+2pl+p+1}$. 

Using these operators, we define the strategy. The provers' strategy $(|\overline{\varphi}\rangle,\{P_k\})$ for $\cal{P}$ constructed from the LHI protocol is as follows. (The strategy $(|\overline{\varphi}\rangle,\{P_k\})$ means that the state shared initially by the provers is $|\overline{\varphi}\rangle$, at the first turn prover $k$ sends the substate in $\mathsf{M}_k$ of $|\overline{\varphi}\rangle$ and at the third turn prover $k$  applies $P_k$ on $(\mathsf{P}_k,\mathsf{M}_k,\mathsf{Me}_k)$ and sends back $\mathsf{M}_k$.)
\begin{equation}
\begin{split}
&|\overline{\varphi}\rangle:=\frac{\Pi_{0^n} W |\varphi\rangle}{\|\Pi_{0^n} W |\varphi\rangle\|},\\
&P_k:=A'_kB_kA_k.\\
\end{split}
\label{philine}
\end{equation}

Next, we calculate the acceptance probability of the strategy $(|\overline{\varphi}\rangle,\{P_k\})$.
Now we take $|\varphi\rangle$ as follows;
\begin{equation}
|\varphi\rangle=\Sigma_{t\in[0,T]} |t\rangle|\psi_t\rangle+|\bot\rangle
\label{varphi}
\end{equation}
where $|\psi_t\rangle$ is an (unnormalized) vector, and $|\bot\rangle$ is a state that is rejected by any Hamiltonian of (e)(consistency of the time counters) in Figure 5.
$\||\bot\rangle\|$ is bounded by the following inequality.
\begin{equation}
\begin{split}
\ \||\bot\rangle\|
&=\bigl\|\sum_{t\in[T+2,2T]}|01\rangle\langle01|_{\mathsf{C}_{t-T-1},\mathsf{C}_{t-T}}|\bot\rangle\bigr\|\\
&\le\sum_{t\in[T+2,2T]}\||01\rangle\langle01|_{{\mathsf{C}_{t-T-1},\mathsf{C}_{t-T}}}|\bot\rangle\|\\
&=\sum_{t\in[T+2,2T]}\||01\rangle\langle01|_{{\mathsf{C}_{t-T-1},\mathsf{C}_{t-T}}}|\varphi\rangle\|\\
&\le(T-1)\sqrt{(2T+n)\epsilon}\\
&=O(\sqrt{T^3\epsilon})
\end{split}
\label{boteval}
\end{equation}
The first equality follows since $|\bot\rangle$  is an incorrect time register state. The first inequality is the triangle inequality. The second equality follows by the definition of $|\varphi\rangle$ in (\ref{varphi}). The third inequality follows by the probability of rejection by (e). The last equality follows by the assumption $n\le T$ in Figure \ref{LHp}(A). 

We bound the following terms (\ref{eq;vch}), (\ref{eq;com1}), (\ref{eq;com2}), (\ref{eq;pr}), using the assumption that the error probability of the LHI protocol is $\epsilon$ or less.
\begin{align}
&\||\psi_t\rangle-U_t|\psi_{t-1}\rangle\|\ \ (t\in[1,t_0]\cup[t_0+2pl+p+1,T]), \label{eq;vch}\\
&\||\psi_t\rangle-D_{k,j}^\dagger SWAP _{(k,j),\mathsf{Me}_k} D_{k,j}|\psi_{t-1}\rangle\|\ (t\in[t_0+1,t_0+pl],\ t=t_0+(k-1)l+j), \label{eq;com1} \\
&\||\psi_t\rangle-{D'_{k,j}}^\dagger SWAP _{(k,j),\mathsf{Me}_k} D'_{k,j}|\psi_{t-1}\rangle\|\ (t\in[t_0+pl+p+1,t_0+2pl+p],\ t=t_0+pl+p+(k-1)l+j),\label{eq;com2}\\
&\||\psi_t\rangle-E_{k}^\dagger X_{\mathsf{Me}_k} E_{k}|\psi_{t-1}\rangle\|\ \ (t=t_0+pl+k).\label{eq;pr}
\end{align}
We can bound (\ref{eq;vch}) as follows;
\begin{equation}
\label{v-uni}
\begin{split}
\||\psi_{t}\rangle-U_t|\psi_{t-1}\rangle\|
&=\frac{1}{2}(\||\psi_{t}\rangle-U_t|\psi_{t-1}\rangle\|+\||\psi_{t-1}\rangle-U_t^\dagger|\psi_{t}\rangle\|)\\
&\le\frac{\sqrt2}{2}\sqrt{\||\psi_{t}\rangle-U_t|\psi_{t-1}\rangle\|^2+\||\psi_{t-1}\rangle-U_t^\dagger|\psi_{t}\rangle\|^2}\\
&=\frac{\sqrt2}{2}\||t\rangle(|\psi_t\rangle-U_t|\psi_{t-1}\rangle)+|t-1\rangle(|\psi_{t-1}\rangle-U_t^\dagger|\psi_t\rangle)\|\\
&=\frac{\sqrt2}{2}\|H_t|\varphi\rangle\|=O(\sqrt{T\epsilon}).
\end{split}
\end{equation}
The first equality follows from modifying $2\||\psi_{t}\rangle-U_t|\psi_{t-1}\rangle)\|$ to $\||\psi_{t}\rangle-U_t|\psi_{t-1}\rangle)\|$$+\||\psi_{t}\rangle-U_t|\psi_{t-1}\rangle)\|$ and applying $U_t^\dagger$ to the second term. The first inequality is the Cauchy-Schwarz inequality. The second equality follows since $|t\rangle$ and $|t-1\rangle$ are orthogonal vectors. The third equality follows from the definitions of $H_t$ in Figure \ref{LHp}. The final equality follows since the whole rejection probability is at most $\epsilon$ and  the verifier selects $H_t$ with probability $\Omega(1/T)$.

We can bound (\ref{eq;com1}) as follows. (\ref{eq;com2}) is bounded similarly;
\begin{equation}
\label{com-uni}
\begin{split}
&\||\psi_{t}\rangle-D_{k,j}^\dagger SWAP _{(k,j),\mathsf{Me}_k} D_{k,j}|\psi_{t-1}\rangle\|\\
&=\frac{1}{2}(\|D_{k,j}|\psi_t\rangle-SWAP _{(k,j),\mathsf{Me}_k} D_{k,j}|\psi_{t-1}\rangle\|+\|SWAP _{(k,j),\mathsf{Me}_k}D_{k,j}|\psi_{t}\rangle-D_{k,j}|\psi_{t-1}\rangle\|)\\
&\le\frac{\sqrt2}{2}\sqrt{\begin{split}&\|D_{k,j}|\psi_t\rangle-SWAP _{(k,j),\mathsf{Me}_k} D_{k,j}|\psi_{t-1}\rangle\|^2+\|SWAP _{(k,j),\mathsf{Me}_k}D_{k,j}|\psi_{t}\rangle-D_{k,j}|\psi_{t-1}\rangle\|^2\end{split}}\\
&=\frac{\sqrt2}{2}\||t\rangle(D_{k,j}|\psi_t\rangle-SWAP _{(k,j),\mathsf{Me}_k} D_{k,j}|\psi_{t-1}\rangle)+|t-1\rangle(D_{k,j}|\psi_{t-1}\rangle-SWAP _{(k,j),\mathsf{Me}_k}D_{k,j}|\psi_t\rangle)\|\\
&=\frac{\sqrt2}{2}\|| H_tD_{k,j}|\varphi\rangle\|=O(\sqrt{T\epsilon}).
\end{split}
\end{equation}
The first equality follows from modifying $\||\psi_{t}\rangle-U_t|\psi_{t-1}\rangle)\|$ to $\frac{1}{2}(\||\psi_{t}\rangle-U_t|\psi_{t-1}\rangle)\|+\||\psi_{t}\rangle-U_t|\psi_{t-1}\rangle)\|)$ and applying $D_{k,j}$ to the first term and $SWAP _{(k,j),\mathsf{Me}_k}D_{k,j}$ to the second term. The first inequality is the Cauchy-Schwarz, the second equality follows since $|t\rangle$ and $|t-1\rangle$ are orthogonal vectors. The third equality follows from the definition of  $H_t$ of (b). The final equality follows since the whole rejection probability is $\epsilon$ or less and  the verifier selects $H_t$ with probability $\Omega(1/T)$.

We can bound (\ref{eq;pr}) as follows;
\begin{equation}
\label{pr-uni}
\begin{split}
\||\psi_t\rangle-E_k^\dagger X_{\mathsf{Me}_k}E_k|\psi_{t-1}\rangle\|
&=\sqrt{2}\| |-\rangle\langle-|\bigl(|0\rangle_{\mathsf{C}_t} (|\psi_t\rangle-E_k^\dagger X_{\mathsf{Me}_k}E_k|\psi_{t-1}\rangle)\bigr)\|\\
&=\sqrt{2}\||-\rangle\langle-|_{\mathsf{C}_t}\bigl(\sqrt2 |+\rangle_{\mathsf C_t}X_{\mathsf{Me}_k}E_k|\psi_t\rangle -|0\rangle_{\mathsf C_t} (X_{\mathsf{Me}_k}E_k|\psi_{t}\rangle-E_k|\psi_{t-1}\rangle)\bigr)\|\\
&=\sqrt{2}\||-\rangle\langle-|_{\mathsf{C}_t}(|1\rangle_{\mathsf C_t}X_{\mathsf{Me}_k}E_k|\psi_t\rangle+|0\rangle_{\mathsf C_t} E_k|\psi_{t-1}\rangle)\|\\
&=\sqrt{2}\||-\rangle\langle-|_{\mathsf{C}_t} (|t\rangle X_{\mathsf{Me}_k}E_k|\psi_{t}\rangle+|t-1\rangle E_k|\psi_{t-1}\rangle)\|\\
&=\sqrt{2}\||-\rangle\langle-|_{\mathsf{C}_t} \bigl((CNOT_{\mathsf{C}_t,\mathsf{Me}_k}(|t\rangle(E_k|\psi_{t}\rangle+|t-1\rangle E_k|\psi_{t-1}\rangle)\bigr)\|\\
&=\sqrt{2}\|H_tE_k|\varphi\rangle\|=O(\sqrt{T\epsilon}).
\end{split}
\end{equation}
The first equality follows from adding register $\mathsf{C}_t$. The second equality follows from multiplying $|0\rangle_{\mathsf{C}_t} (|\psi_t\rangle-E_k^\dagger X_{\mathsf{Me}_k}E_k|\psi_{t-1}\rangle)$ by  $-1$, applying $X_{\mathsf{Me}_k}E_k$ and $\langle-|+\rangle=0$. The third equality follows from $|+\rangle=\frac{|0\rangle+|1\rangle}{\sqrt2}$. The forth equality follows from $|t-1\rangle=|1\cdots1\rangle_{\mathsf{C}_1,\cdots,\mathsf{C}_{t-1}}|0\cdots 0\rangle_{\mathsf{C}_t,\cdots ,\mathsf{C}_{T+2pl+p}}$ and $|t\rangle=|1\cdots 1\rangle_{\mathsf{C}_1,\cdots,\mathsf{C}_t}|0\cdots 0\rangle_{\mathsf{C}_{t+1},\cdots,\mathsf{C}_{T+2pl+p}}$. The fifth equality follows from that $CNOT_{\mathsf{C}_t,\mathsf{Me}_k}$ acts as $X_{\mathsf{Me}_k}$ if and only if register $\mathsf{C}_t$ is $|1\rangle$. The last equality is the definition of $H_t$ of (c) and $E_k$. 

The following equations evaluate the norm of $\Pi_{0^n} W |\varphi\rangle$.
\begin{equation}
\label{norm}
\begin{split}
&\!\!\!\! \|\Pi_{0^n}W|\varphi\rangle\|\\
&\ge1-\|({\rm Id}-\Pi_{0^n})W|\varphi\rangle\|\\
&=1-\|({\rm Id}-\Pi_{0^n})(\sum_{t=0}^T |t\rangle\langle t| \otimes U_1^\dagger\cdots U_t^\dagger+\Pi_{bad time}\otimes \mathrm{Id})(\sum_{t=0}^T|t\rangle|\psi_t\rangle+|\bot\rangle)\|\\
&=1-\sum_{t=0}^T\Bigl(\|({\rm Id}-\Pi_{0^n})|t\rangle\bigl(|\psi_0\rangle+\sum_{t'=1}^t (U_1^\dagger\cdots U_{t'}^\dagger|\psi_{t'}\rangle-{U_1}^\dagger\cdots U_{t'-1}^\dagger|\psi_{t'-1}\rangle)\bigr)\|\Bigr)-O(\sqrt{T^3\epsilon})\\
&\ge1-\sum_{t=0}^T\Bigl(\|({\rm Id}-\Pi_{0^n})|\psi_0\rangle\|+\sum_{t'=1}^t \|({\rm Id}-\Pi_{0^n})(U_1^\dagger\cdots U_{t'}^\dagger|\psi_{t'}\rangle-{U_1}^\dagger\cdots U_{t'-1}^\dagger|\psi_{t'-1}\rangle\bigr)\|\Bigr)-O(\sqrt{T^3\epsilon}).\\
\end{split}
\end{equation}
Here, we ignore the term $\sum_{t'=1}^t$ for $t=0$ for the ease of notation. The first inequality follows from the triangle inequality. The first equality follows from the definition of $W$ in (\ref{opr}) and of $|\varphi\rangle$ in (\ref{varphi}). The second equality follows from $U_1^\dagger\cdots U_t^\dagger|\psi_t\rangle=|\psi_0\rangle +\Sigma_{t'=1}^t(U_1^\dagger\cdots U_{t'}^\dagger|\psi_{t'}\rangle-{U_1}^\dagger\cdots U_{t'-1}^\dagger|\psi_{t'-1}\rangle)$ and $\||\bot\rangle\|=O(\sqrt{T^3\epsilon})$ by (\ref{boteval}). The second inequality is the triangle inequality. 

The term $\|(\mathrm{Id}-\Pi_{0^n})|\psi_0\rangle\|$ in the last line of (\ref{norm}) is bounded by test (f) in Figure \ref{LHp} as follows;
\begin{equation}
\label{term1}
\|(\mathrm{Id}-\Pi_{0^n})|\psi_0\rangle\|\le\sum_{t\in[2T+1,2T+n]}\|(\mathrm{Id}-\Pi_0^{t-2T})|\psi_0\rangle\| \le O(n\sqrt{T\epsilon}).
\end{equation}
The second inequality  follows since each term in the summation of (\ref{term1}) is the rejection probability of each $t$ in test (f) in Figure \ref{LHp}.
 
The terms $\|(\mathrm{Id}-\Pi_{0^n})(U_T\cdots U_{t'+1}|\psi_{t'}\rangle-U_T\cdots U_{t'}|\psi_{t'-1}\rangle)\|$ in the same line of (\ref{norm}) are bounded as follows; 
\begin{equation}
\label{term2}
\|(\mathrm{Id}-\Pi_{0^n})(U_T\cdots U_{t'+1}|\psi_{t'}\rangle-U_T\cdots U_{t'}|\psi_{t'-1}\rangle)\| \le \||\psi_{t'}\rangle-U_{t'}|\psi_{t'-1}\rangle\|\le O(\sqrt{T\epsilon}).
\end{equation}
The first inequality follows by omitting the projection $(\mathrm{Id}-\Pi_{0^n})$ and applying ${U_{t'+1}}^\dagger\cdots {U_T}^\dagger$. The second inequality follows since the term $\||\psi_{t'}\rangle-U_{t'}|\psi_{t'-1}\rangle\|$ equals to one of (\ref{eq;vch},\ref{eq;com1},\ref{eq;com2},\ref{eq;pr}), and (\ref{eq;vch},\ref{eq;com1},\ref{eq;com2},\ref{eq;pr}) are bounded by (\ref{v-uni},\ref{com-uni},\ref{pr-uni}).

From (\ref{norm},\ref{term1},\ref{term2}) and the assumption $n\le T$, we have the following estimation;
\begin{equation}
\|\Pi_{0^n}W|\varphi\rangle\|=1-O(nT\sqrt{T\epsilon})-O(T^2\sqrt{T\epsilon})=1-O(\sqrt{T^5\epsilon}).
\label{linephi}
\end{equation}
From this, the next estimation follows; 
\begin{equation}
\label{norm2}
\|(\mathrm{Id}-\Pi_{0^n}) W |\varphi\rangle\|^2=1-\|\Pi_{0^n}W|\varphi\rangle\|^2=O(\sqrt{T^5\epsilon}).
\end{equation}

Finally, we bound the rejection probability $p_{rej}$.  Let $\Pi_{rej}$ be the projection onto the rejection of $\cal{P}$. Then,
\begin{equation}
\label{rej}
\begin{split}
\sqrt{p_{rej}}&=\|\Pi_{rej}U_{T}\cdots U_1|\overline{\varphi}\rangle\|\\
&=(1+O(\sqrt{T^5\epsilon}))\|\Pi_{rej}U_T\cdots U_1\Pi_{0^n}W|\varphi\rangle\|\\
&\le(1+O(\sqrt{T^5\epsilon}))(\|\Pi_{rej}U_T\cdots U_1W|\varphi\rangle\|+\|\Pi_{rej}U_T\cdots U_1({\rm Id}-\Pi_{0^n})W|\varphi\rangle\|)\\
&\le(1+O(\sqrt{T^5\epsilon}))\bigl(\|\Pi_{rej}U_T\cdots U_1W|\varphi\rangle\|+O(\sqrt[4]{T^5\epsilon})\bigr)\\
&=\|\Pi_{rej}\sum_{t=0}^T U_T\cdots U_{t+1}|\psi_t\rangle\|+O(\sqrt{T^3\epsilon})+O(\sqrt[4]{T^5\epsilon})\\
&\le\sum_{t=0}^T\Bigl(\|\Pi_{rej}|\psi_T\rangle\|+\sum_{t'=t+1}^T \|\Pi_{rej}U_T\cdots U_{t'+1}|\psi_{t'}\rangle-\Pi_{rej}U_T\cdots U_{t'}|\psi_{t'-1}\rangle\|\Bigr)+O(\sqrt[4]{T^5\epsilon}).\\
\end{split}
\end{equation}
Here, we ignore the term $\sum_{t'=t+1}^T$ for $t=T$ for the ease of notation.  The second equality follows from (\ref{linephi}) and the definition of $|\overline\varphi\rangle$ in (\ref{philine}). The first inequality follows from the triangle inequality. The second inequality follows from (\ref{norm2}). The third equality follows from the definition of $W$ in (\ref{opr}), the definition of $|\varphi\rangle$ in (\ref{varphi}), and $\||\bot\rangle\|=O(\sqrt{T^3\epsilon})$  by (\ref{boteval}). The last inequality follows from the triangle inequality. 

The term $\|\Pi_{rej}|\psi_T\rangle\|$ in the last line of (\ref{rej}) is bounded by test (d) in Figure \ref{LHp} as follows; 
\begin{equation}
\label{term3}
\|\Pi_{rej}|\psi_T\rangle\|\le O(\sqrt{T\epsilon}).
\end{equation}

The terms $\|\Pi_{rej}U_T\cdots U_{t'+1}|\psi_{t'}\rangle-\Pi_{rej}U_T\cdots U_{t'}|\psi_{t'-1}\rangle\|$ in the same line can be bounded by (\ref{v-uni},\ref{com-uni},\ref{pr-uni}) as follows; 
\begin{equation}
\label{term4}
\|\Pi_{rej}U_T\cdots U_{t'+1}|\psi_{t'}\rangle-\Pi_{rej}U_T\cdots U_{t'}|\psi_{t'-1}\rangle\| \le \||\psi_{t'}\rangle-U_{t'}|\psi_{t'-1}\rangle\|\le O(\sqrt{T\epsilon}).
\end{equation} 
From (\ref{rej},\ref{term3},\ref{term4}), $\sqrt{p_{rej}}\le O(\sqrt[4]{T^5\epsilon})$ follows.
\end{proof}

\subsection{Addition of the GHZ test}
Second, we construct the LHI+ protocol: we add the GHZ test which is an analogue of the GHZ test of Theorem 5 to the LHI protocol, which is described in Figure \ref{int}. The honest verifier of the LHI protocol sends all provers (except prover 0) the same bits, but the malicious verifier may send different bits. To prevent this attack, the provers control which term $H_t$ is measured by the verifier in Figure \ref{LHp}(A). To avoid also that the provers are malicious, we use the GHZ test, which guarantees that the provers really chooses $H_t$ uniformly at random. The reason why we do not use the coin-flipping protocol to decide $t$ is that we do not know any multi-parity coin-flipping protocol  among the provers and the verifier. For example, if we use the two-party coin-flipping protocol between each of the provers and the verifier, the malicious verifier could choose different $t$ depending on the provers. The analysis of the LHI+ protocol is almost the same as the proof of Theorem 4.

\begin{lemm}If the completeness/soundness  of the protocol in Figure \ref{LHp} are  $1-\epsilon/1-\delta$, then the completeness/soundness of the LHI+ protocol in Figure \ref{int} are $1-\frac{\epsilon}{2}$/$1-\frac{\delta^2}{8}$.
\end{lemm}

\begin{proof}
The completeness is obvious. We show the soundness. The analysis is almost the same as the proof of Theorem 5. The only differences are the definition of $\Pi_{GHZ}$ and $\Pi_{rej}$. Let $U_g^i$ be the unitary operator that prover $i$ uses for the GHZ test and $U_h^i$ be the unitary operator that prover $i$ uses for the history test for $i=1,...,p$.

Let $\rho'':=\rm{Tr}_{GHZ}$$(\Pi_{GHZ}(\otimes_{i=1}^p U_g^i)\rho(\otimes_{i=1}^p {U_g^i}^\dagger)\Pi_{GHZ})$. Here, $\rho$ is the initial state of the provers in the LHI+ protocol, $\Pi_{GHZ}$ is the projection onto $(\frac{|0^{p+1}\rangle+|1^{p+1}\rangle}{\sqrt2})^{\otimes u}$ of ($\mathsf{G}_0, \mathsf{G}_1, \mathsf{G}_2,...,\mathsf{G}_p$), and $\rm Tr_{GHZ}$ means the traceout of registers ($\mathsf{G}_0, \mathsf{G}_1, \mathsf{G}_2,...,\mathsf{G}_p$). Denote $\rm Tr$$\rho''=1-\epsilon_1$, where $\epsilon_1$ is the rejection probability of the GHZ test. Let $\frac{\rho''}{1-\epsilon_1}:=\rho'$ be the initial state of provers. Let $\epsilon_2$ be the rejection probability of the history test.

Now we construct the strategy of  the LHI protocol from the strategy of the LHI+ protocol. The private registers of prover $i$ are $\mathsf{P}_i,\mathsf{G}_i$, for $i=1,...,p$. Initially prover 0 has only register $\mathsf{V}$. The initial state in ($\mathsf{P}_1,...,\mathsf{P}_p,\mathsf{M}_1,...,\mathsf{M}_p$) is $\rho'$. If the verifier sends $b\in\{0,1\}^u$ to prover $i$ ($i=1,...,p$), prover $i$ sets $|b\rangle$ in register $\mathsf{G}_i$, applies $U^i_h{U^i_g}^\dagger$ and sends $\mathsf{M}_i$ to the verifier. Prover 0 always sends $\mathsf{V}$. The operations by provers are as follows;

\begin{equation}\rho'\rightarrow(\otimes_{i=1}^p U_h^i{U_g^i}^\dagger) ((|b\rangle\langle b|)^{\otimes p}\otimes \rho') (\otimes_{i=1}^p U_g^i{U_h^i}^\dagger).
\end{equation}
The rejection probability of the LHI protocol $p_{rej,lh}$ is bounded by the rejection probability $p_{rej}$  of the LHI+ protocol as follows;
\begin{equation}
\begin{split}
p_{rej,lh}&:=\frac{1}{2^u}\sum_{b\in\{0,1\}^u}\mathrm{Tr}[\Pi_{rej}(\otimes_{i=1}^p U_h^i {U_g^i}^\dagger) ((|b\rangle\langle b|)^{\otimes p}\otimes \rho') (\otimes_{i=1}^p U_g^i{U_h^i}^\dagger)]\\
&=\frac{1}{1-\epsilon_1}\mathrm{Tr} [\Pi_{rej}(\otimes_{i=1}^p  U_h^i{U_g^i}^\dagger)\Pi_{GHZ}(\otimes_{i=1}^p U_g^i)\rho (\otimes_{i=1}^p {U_g^i}^\dagger)\Pi_{GHZ}(\otimes_{i=1}^p U_g^i{U_h^i}^\dagger)]\\
&\le \sqrt\epsilon_1+\mathrm{Tr} [\Pi_{rej}(\otimes_{i=1}^p  U_h^i{U_g^i}^\dagger)(\otimes_{i=1}^p U_g^i)\rho (\otimes_{i=1}^p {U_g^i}^\dagger)(\otimes_{i=1}^p {U_g^i}{U_h^i}^\dagger)]\\
&= \sqrt\epsilon_1+\mathrm{Tr}[\Pi_{rej}(\otimes_{i=1}^p U_h^i)\rho(\otimes_{i=1}^p {U_h^i}^\dagger)]=\sqrt\epsilon_1+\epsilon_2\le\sqrt\epsilon_1+\sqrt\epsilon_2\le2\sqrt {2p_{rej}}.
\end{split}
\end{equation}

The first equality follows from that $\frac{1}{2}(|0^p\rangle\langle0^p|+|1^p\rangle\langle1^p|)$ equals to the $p$-qubits substate of the GHZ state. The first inequality follows from that $\|\rho'-\rho\|\le\sqrt\epsilon_1$ by Lemma \ref{2PO}.

\end{proof}
\begin{figure}
\hrulefill\\
LHI+ protocol.\\\\
We can assume $2^u=n+2T$ for some integer $u$ without loss of generality. Let $\mathsf{G}_i$ $(i=0,...,p)$ be registers each of which consists of $u$ qubits.
\begin{enumerate}
\item
Select $b\in\{0,1\}$ uniformly at random. Send $b$ to provers $1,...,p$ and nothing to prover 0. 
\item
Do the following test depending on $b=0,1$;
\begin{enumerate}
\item[(1)]
($b=0$: GHZ test)\\
(1.1)  Prover $i$ ($i=1,...,p$) sends $\mathsf{G}_i$. Prover 0 sends $(\mathsf{V}, \mathsf{M}_1,..., \mathsf{M}_p$, $\mathsf{C}_1,...,\mathsf{C}_T)$ and $\mathsf{G}_0$.\\
(1.2) Measure $\mathsf{G}_0,\mathsf{G}_1,...,\mathsf{G}_p$ by the projection onto  $\frac{1}{\sqrt2}(|0^{p+1}\rangle+|1^{p+1}\rangle)^{\otimes u}$. If $(\frac{1}{\sqrt2} (|0^{p+1}\rangle+|1^{p+1}\rangle))^{\otimes u}$ is measured, then accept. Otherwise reject.
\item[(2)]
($b=1$: History test)\\
(2.1) Prover $i$ ($i=1,...,p$) sends $\mathsf{Me}_i$ and $\mathsf{G}_i$. Prover 0  sends $(\mathsf{V}, \mathsf{M}_1,..., \mathsf{M}_p$, $\mathsf{C}_1,...,\mathsf{C}_T)$ and $\mathsf{G}_0$.\\
(2.2) Measure $\mathsf{G}_0$ in the computational basis. Let $t$ be the outcome. Measure by $H_t$ in Figure \ref{LHp}. If the outcome is accepted by the protocol in Figure 5, then accept. Otherwise reject.
\end{enumerate}
\end{enumerate}
\hrulefill
\caption{LHI+ protocol. Note that we add one prover and the  GHZ test to the LHI protocol.}
\label{int}
\end{figure}
\subsection{Sketch of the technique of \cite{BJSW} and why it works for our protocol\label{Sketch}}
Finally, we give the zero-knowledge protocol for $\sf QMIP^*$ based on the LHI+ protocol, following the zero-knowledge protocol by Broadbent et al. \cite{BJSW}. In this subsection, we roughly sketch the quantum zero-knowledge protocol by Broadbent et al.\cite{BJSW} and why the technique of \cite{BJSW} works for our protocol.
The result of \cite{BJSW} consists of following ingredients;
\begin{enumerate}
\item The verifier tries to verify a  restricted form of the Local Hamiltonian problem, called the Clifford Hamiltonian problem, which is shown to be $\sf{QMA}$ hard. To this end, the verifier only has to measure only Clifford Hamiltonians.
\item The honest prover encodes the witness of the Clifford Hamiltonian problem by a CSS code \cite{NS1}, a quantum one time pad and a permutation. The quantum one time pad and the permutation are the secret key of the encoding.
\item The prover sends the encoded witness and the commitment of the key of the encoding. The verifier measures the encoded witness by one of Hamiltonians and sends the output to the prover. The prover proves that the output of the verifier's measurement corresponds to a yes output of the original Clifford Hamiltonian problem by a zero-knowledge protocol for NP \cite{W2}. This correspondence critically uses the restriction on Hamiltonians and  the transversality for Clifford gates which is a characteristic of CSS codes. 
\item The malicious verifier's circuits can be replaced by simulators by the assumption of the existence of commitment schemes, and after the replacement of the verifier's circuits, the witness state can be replaced by a state preparable in polynomial time. 
\end{enumerate}
If the malicious verifier of the LHI protocol can send only the honest query, the analogues of the above items in our case are as follows;
\begin{enumerate}
\item The verifier of the LHI protocol also measures by only Clifford Hamiltonians.
\item The encoding step of each prover does not depend on the other provers, and hence similar encoding can be done.
\item The corresponding step can be done by one of the provers directly.
\item If the query of the (malicious) verifier is honest and the qubits sent from the provers is honest, this step also can be done directly.
\end{enumerate}
As we cannot assume in general that the malicious verifier of the LHI protocol can send only the honest query,  we do not directly use the LHI protocol but the LHI+ protocol, which adds the GHZ test.

\subsection{Final zero-knowledge protocol: the technique of Broadbent et al. \cite{BJSW}\label{Fin;sec}}
We construct the zero-knowledge protocol based on the protocol in Figure \ref{int}. The technique is almost the same as that of Broadbent et al. \cite{BJSW} and the analysis is also almost the same. In this paper, we construct the protocol and explain why the technique of \cite{BJSW} works. The summary of the protocol is given in Figure \ref{Final} and we describe the protocol here.

Let $\mathsf{X}=(\mathsf{X}_1,...,\mathsf{X}_{i_0+p})$ be $(i_0+p)$ registers each of which has 1 qubit, where $i_0$ is the number of qubits of registers ($\mathsf{V},\mathsf{M}_1,...,\mathsf{M}_p,\mathsf{C}_1,...,\mathsf{C}_T$) in Figure \ref{int}. Prover 0 has $(\mathsf{X}_1,...,\mathsf{X}_{i_0})$ that corresponds to ($\mathsf{V},\mathsf{M}_1,...,\mathsf{M}_p,\mathsf{C}_1,...,\mathsf{C}_T$). For $i=1,...,p$, prover $i$ has $\mathsf{X}_{i_0+i}$ that corresponds to register $\mathsf{Me}_i$ in Figure \ref{int}. 
\subsubsection{Verifier's message\label{Vmes}}
First, the verifier selects $\overline{b}\in\{0,1\}$ uniformly at random and sends $\overline{b}$ to provers $1,...,p$. Note that $\overline{b}=0$ corresponds to the  GHZ test and $\overline{b}=1$ corresponds to the history test.  The verifier sends nothing to prover 0.
\subsubsection{Provers' encoding\label{Penc}}
Prover $i$ $(1\le i\le p)$ who received $\overline{b}=0$ sends $\mathsf{G}_i$. Prover $i$ who received $\overline{b}=1$ and prover 0 encode $\mathsf{X}$ in four steps.  Let $N$ be the length of a concatenated Steane code that the provers use. Here, a concatenated Steane code is a code such that 1 qubit is encoded by the 7-qubit Steane code several times repeatedly. In \cite{BJSW}, $N$ is taken to be an even power of 7 and bounded by a polynomial in instance size.
\begin{enumerate}
\item For $i=1,...,i_0+p$, $\mathsf{X}_i$ is encoded by the concatenated Steane code to the $N$ qubit state $(\mathsf{Y}_1^i,...,\mathsf{Y}_N^i)$.
\item For each $i$, the provers concatenate additional $N$ qubits to  $(\mathsf{Y}_1^i,...,\mathsf{Y}_N^i)$, each of which is chosen from $|0\rangle,|+\rangle$,$|\!\!\circlearrowright\rangle(:=\frac{1}{\sqrt{2}}(|0\rangle+i|1\rangle))$ uniformly at random. We call these $N$ qubits as trap qubits. At this point, each $\mathsf{X}_i$ is transfered to $2N$ qubits $(\mathsf{Y}_1^i,...,\mathsf{Y}_{2N}^i)$. The provers store the string $r=r_1\cdots r_{i_0+p}$, where $r_i\in\{0,+,\circlearrowright\}^N$, representing the trap qubits.
\item Select a permutation $\pi\in S_{2N}$ uniformly at random. All provers use the same permutation $\pi$. Permute $2N$ qubits $(\mathsf{Y}_1^i,...,\mathsf{Y}_{2N}^i)$ by $\pi$.
\item  Select strings $a=a_1\cdots a_{i_0+p}$, $b=b_1\cdots b_{i_0+p}$, where $a_1,... ,a_{i_0+p},b_1,...,b_{i_0+p}\in\{0,1\}^{2N}$, uniformly at random.  Apply quantum one time pad on $(\mathsf{Y}_1^i,...,\mathsf{Y}_{2N}^i)$. Namely, apply $X^{a_i}Z^{b_i}$ on  $(\mathsf{Y}_1^i,...,\mathsf{Y}_{2N}^i)$, where for strings $a_i=a^1_i\cdots a^{2N}_i,b^1_i\cdots b^{2N}_i$, $X^{a_i}Z^{b_i}$ means $X^{a^1_i}Z^{b^1_i}\otimes\cdots\otimes X^{a^{2N}_i}Z^{b^{2N}_i}$. Each of provers $0,2,...,p$ sends each encoded qubits and GHZ qubits. Prover 1 sends his encoded qubits, GHZ qubits and the commitment $z=commit(\pi,a,b,s)$. Here, $commit(\pi,a,b,s)$ is a string to commit the string $(\pi,a,b)$  with a random string $s$ using the commitment scheme assumed in Theorem 2.
\end{enumerate}
\begin{figure}
\hrulefill\\
Preparation \\\\
 Provers $0,1,...,p$ select and share a tuple $(r,\pi,a,b)$ uniformly at random, where $r=r_1\cdots r_{i_0+p}$ for $r_1,..., r_{i_0+p}\in\{0,+,\circlearrowright\}^N$, $\pi\in S_{2N}$, and $a=a_1\cdots a_{i_0+p}$, $b=b_1\cdots b_{i_0+p}$ for $a_1,... ,a_{i_0+p},b_1,...,b_{i_0+p}\in\{0,1\}^{2N}$.
The provers will use this random string to encode their qubits which correspond to the qubits that the provers in the LHI+ protocol would send. The encoding process is  described in Section \ref{Penc}.\\
\hrulefill\\
Protocol\\
Select $\overline{b}=0,1$ uniformly at random. Send $\overline{b}$ to prover $1,...,p$, and nothing to prover 0.

Do the following test depending on $\overline{b}=0,1$:
\begin{enumerate}
\item($\overline{b}=0$: GHZ test)
\begin{enumerate}
\item[(1.1)]  Prover $i$ ($i=1,...,p$) sends $\mathsf{G}_i$. Prover 0 encodes $(\mathsf {X}_0,...,\mathsf{X}_{i_0})$
$(=($$\mathsf{V}, \mathsf{M}_1,..., \mathsf{M}_p,\mathsf{C}_1,...,\mathsf{C}_T))$ as described in Section \ref{Penc} and sends these qubits and $\mathsf{G}_0$.
\item[(1.2)] Measure $\mathsf{G}_0,\mathsf{G}_1,...,\mathsf{G}_p$ by the projection onto  $\frac{1}{\sqrt2}(|0^{p+1}\rangle+|1^{p+1}\rangle)^{\otimes u}$. If $\frac{1}{\sqrt2} (|0^{p+1}\rangle+|1^{p+1}\rangle)^{\otimes u}$ is measured, then accept. Otherwise reject.
\end{enumerate}
\item($\overline{b}=1$: History test)
\begin{enumerate}
\item[(2.1)] Prover $i$ ($i=1,...,p$) encodes $\mathsf{X}_{i_0+i}(=\mathsf{Me}_i)$  described in Section \ref{Penc} and sends these qubits and $\mathsf{G}_i$. Prover 1 additionally sends $z=commit(\pi,a,b,s)$. Prover 0 encodes $\mathsf {X}_0,...,\mathsf{X}_{i_0}(=(\mathsf{V},\mathsf{M}_1,..., \mathsf{M}_p,\mathsf{C}_1,...,\mathsf{C}_T))$ as described in Section \ref{Penc} and sends these qubits and $\mathsf{G}_0$.
\item[(2.2.1)] The verifier measures $\mathsf{G}_0$ in the computational basis. Let $t$ be the outcome. 
\item[(2.2.2)] Prover 1  and the verifier engage in a coin-flipping protocol, choosing a two bit string $v$ uniformly at random. Here, $v$ specifies one of the Clifford gates of $H_t=\sum_v H'_{t,v}$. $H'_{t,v}$ is described in Section \ref{Vmea}.
\item[(2.2.3)]  The verifier applies the Clifford operation $C_{t,v}$ transversally to the qubits $(\mathsf{Y}_1^{i_1}$,...,$\mathsf{Y}_{2N}^{i_1})$, $(\mathsf{Y}_1^{i_2},...,\mathsf{Y}_{2N}^{i_2})$,..., $(\mathsf{Y}_1^{i_k},...,\mathsf{Y}_{2N}^{i_k})$ as described in Section \ref{Vmea} and measures all of these qubits in the computational basis, for $(i_1,...,i_k)$ being the indices of the qubits upon which the Hamiltonian term $H_{t,v}$ acts non-trivially. The verifier sends the output to prover 1.
\item[(2.2.4)] Prover 1 checks whether the output sent from the verifier is consistent with the trap qubits and Steane code (described in Section \ref{Pche}). If they are inconsistent, then abort.  If they are consistent, prover 1 proves that the output corresponds to a yes output of the LHI+ protocol by a zero-knowledge protocol of NP.
\end{enumerate}
\end{enumerate}
\hrulefill\\
\caption{Summary of Zero-Knowledge Protocol for $\sf QMIP^*$}
\label{Final}
\end{figure}
\subsubsection{Verifier's measurement\label{Vmea}}
If the verifier sends $\overline{b}=0$ (the GHZ test), then he/she measures $\mathsf{G}_0,...,\mathsf{G}_p$ as in the LHI+ protocol.

If the verifier sends $\overline{b}=1$ (the history test), then the verifier measures the state received from provers as follows: $H_t$ is the summation of at most four Clifford gates, that is, $H_t=\sum_v H'_{t,v}$, where $H'_{t,v}$ is a Clifford gate and $v\in\{00,01,10,11\}$ (see Section \ref{CliH} in this paper and Section 2 in \cite{BJSW}). First, the verifier decides $t$ by measuring in the computational basis on $\mathsf{G}_0$ and $v$ by the coin-flipping protocol  with prover 1 (i.e., the honest prover 1 commits to random $s_1,s_2 \in\{0,1\}$, the honest verifier selects $s'_1,s'_2\in \{0,1\}$ at random, the prover reveals $s_1,s_2$, and the two participants agree that the random bits are $v_i = s_i \oplus s'_i$ for $i=1,2$). The verifier measures the encoded witness transversally by the projection $H'_{t,v}$. Measuring $\sf G_0$ in the computational basis decides $t$ and coin-flipping decides $v$. Here, the measurement by $H'_{t,v}$ means applying the corresponding Clifford gate $C_{t,v}$ to $(\mathsf{Y}_1^{i_1}$,...,$\mathsf{Y}_{2N}^{i_1})$, $(\mathsf{Y}_1^{i_2},...,\mathsf{Y}_{2N}^{i_2})$,..., $(\mathsf{Y}_1^{i_k},...,\mathsf{Y}_{2N}^{i_k})$ and applying the projection onto computational basis. Here, $i_1,\ldots i_k$ is the indices of qubits on which $H'_{t,v}$ non-trivially acts. The output is a string in $\{0,1\}^{2kN}$. Denote the output $u=u_{i_1}\cdots u_{i_k}$, where $u_{i_1},...,u_{i_k}\in\{0,1\}^{2N}$.
\subsubsection{Provers' Check and Response\label{Pche}}
The verifier sends $k$ strings in $\{0,1\}^{2N}$ to prover 1. There uniquely exist strings $y_i,z_i\in\{0,1\}^N$ such that $\pi (y_i z_i)=u_i$ for each $u_i\in\{0,1\}^N$. Prover 1 continues if the following conditions are satisfied.
\begin{enumerate}
\item For all $i\in\{i_1,\ldots,i_k\}$, $y_i\in D_N$, and  at least one $i\in\{i_1,\ldots,i_k\}$, $y_i\in D_N^1$, where $D_N$ is a subset of $\{0,1\}^N$ representing classical codewords of the concatenated Steane code and $D_N^1$ is a subset of $D_N$ corresponding to logical bit 1.
\item $\langle z_{i_1}\cdots z_{i_k}|C_{t,v}^{\otimes N}|r_{i_1}\cdots r_{i_k}\rangle\neq0$.
\end{enumerate}
We define the predicate $R_{t,v}(r,u,\pi)$ which takes the value 1 iff the above two conditions hold. Assume that $R_{t,v}(r,u,\pi)=1$ and prover 1 continues the protocol. For any $a=a_1\cdots a_{i_0+p},b=b_1\cdots b_{i_0+p}$, there uniquely exist $\alpha\in\{\pm i, \pm1\}$ and  $c_1,...,c_{i_0+p},d_1,...,d_{i_0+p}\in\{0,1\}^{2N}$ such that the next equation holds and can be computed in polynomial time in $N$.
\begin{equation}
\label{pred}
C_t^{\otimes 2N}(X^{a_1}Z^{b_1}\otimes\cdots\otimes X^{a_{i_0+p}}Z^{b_{i_0+p}})=\alpha(X^{c_1}Z^{d_1}\otimes\cdots\otimes X^{c_{i_0+p}}Z^{d_{i_0+p}}) C_t^{\otimes 2N} 
\end{equation}
That is, the following statement is a NP statement: there are a string $s$ and a tuple $(\pi, a,b,r)$ such that $commit(\pi,a,b,s)=z$ and $R_{t,v}(r,u\oplus c_1,...,c_{i_0+p},\pi)=1$, where $c$ is defined by Eq.(\ref{pred}). Prover 1 convinces the verifier of this statement by a zero-knowledge protocol of NP.

\subsection{Analysis}
As we mentioned before, the analysis is almost the same as \cite{BJSW}. Hence we explain only the main difference.
\subsubsection{Clifford gates and Clifford Hamiltonians\label{CliH}}
The zero-knowledge protocol by Broadbent et al. \cite{BJSW} critically uses the condition that all Hamiltonians consist of  Clifford gates and projections onto  computational basis.  Here we prove that the following Hamiltonians are the sums of Clifford Hamiltonian for $U_t=SWAP, CNOT$. $SWAP$ and $CNOT$ can be easily constructed by the product of unitary operators of  $\{H\otimes H ,\Lambda(P)\}$. Hence this step is not essentially necessary, but we prove this to simplify the LHI+ protocol. Now $H_t$ is defined as follows;
\begin{equation*}
H_t=|10\rangle_{\mathsf{C}_{t-1},\mathsf{C}_{t+1}}\langle10|_{\mathsf{C}_{t-1},\mathsf{C}_{t+1}}\otimes(|0\rangle\langle0|_{\mathsf{C}_t}\otimes \mathrm{Id}+|1\rangle\langle1|_{\mathsf{C}_t}\otimes \mathrm{Id}-|1\rangle\langle0|_{\mathsf{C}_t}\otimes U_t-|0\rangle\langle1|_{\mathsf{C}_t}\otimes U_t^\dagger)\end{equation*}
$H_t$ acts on $|\rangle_{\mathsf{C}_{t-1},\mathsf{C}_{t+1}}$ as a trivial projection onto computational basis, and hence we consider only other three qubits. For $U_t=SWAP$, $H_t$ is the sum of the projections onto the following vectors;\\\\
\{$|+\rangle(|01\rangle-|10\rangle)$, $|-\rangle|00\rangle$, $|-\rangle|11\rangle$, $|-\rangle(|01\rangle+|10\rangle)$\}.\\\\
Here, $|+\rangle(|01\rangle-|10\rangle)=|+\rangle ((\mathrm{Id} \otimes X) CNOT|-0\rangle)$, and $|-\rangle(|01\rangle+|10\rangle)=|-\rangle ((\mathrm{Id} \otimes X) CNOT|+0\rangle)$ and hence these projections are Clifford Hamiltonians. The control qubits of the CNOT in these operators is the left qubit.

For $U_t=CNOT$, $H_t$ is the sum of the projections onto the following vectors;\\\\
\{$|+1-\rangle$, $|-1+\rangle$, $|-0+\rangle$, $|-0-\rangle$\}.
\subsubsection{Soundness}
In the analysis of Broadbent et al.'s protocol \cite{BJSW}, the prover can prepare the state accepted by the original local Hamiltonian test with high probability  by decoding the encoded qubits which can pass their zero-knowledge protocol with high probability. The decoding process is applied to each logical qubit isolatedly. Hence, if the provers in Figure \ref{Final} pass with high probability,  then the provers can also pass the LHI protocol with high probability by decoding the qubits in Figure \ref{Final}.
\subsubsection{Zero-Knowledge}
Finally, we discuss zero-knowledge. Similarly to the proof of Theorem 4, we only have to consider the case that the malicious verifier requires provers $1,...,p$ to do the history test, and we can assume that prover $i$ who received $\overline{b}=1$ measures $\mathsf{G}_i$ in the computational basis. In the case of the history test, honest provers send the state that depends on the uniform random variable $t\in[{2^u}]$. The state may depend on $t$, but the analysis of \cite{BJSW} showed that there is only negligible change of the outputs of the simulator if the honest measurement on the state can pass the prover's check with high probability.

\section{Conclusion}
There are obvious open problems:  whether there exist the statistical/perfect zero-knowledge systems of $\mathsf{QMIP^*}$. One possible method is the algebraic technique. In quantum complexity theory, the technique of enforcing algebraic structures on the strategy of provers is recently applied to investigate the power of $\mathsf{MIP^*}$  (to prove $\mathsf{NEXP}$ $\subseteq\mathsf{MIP^*}$\cite{IV,NV,Vid}, to construct short proofs for $\sf{QMA}$ with large completeness-soundness gap \cite{NV2}, and to prove $\mathsf{NEXP}$ $\subseteq\mathsf{MIP^*}$ with zero-knowledge \cite{CFGS}, for example). Extensions of such algebraic methods to history states of $\mathrm{QMIP^*}$ protocols may enable perfect zero-knowledge systems for $\mathsf{QMIP^*}$.


The parameters will not be optimal. Though most improvements of parameters will directly follow from improvements of $\mathsf{QMIP^*}$ protocols without zero-knowledge condition, we note an important problem related to zero-knowledge. Parallel repetition is a  direct tool to improve completeness/soundness gap. Parallel repetition of zero-knowledge protocols, however, may not preserve zero-knowledge even in single-prover classical zero-knowledge systems (\cite{Gol}, Section 4.5.4). Hence it might be difficult to improve completeness/soundness gap preserving the number of turns by parallel repetition.

Finally, we believe that LHI protocols of interactive proofs would be a powerful tool which makes much previous research for Local Hamiltonian problems applicable to interactive proof systems. 
\subsection*{Acknowledgments}
The author is grateful to Prof. Nishimura for useful discussion, careful reading and heavy revision of this paper.

\appendix

\section{Elimination of the honest condition without new provers}
In this section, we eliminate the honest condition without adding extra provers. We use the technique that Kobayashi \cite{Ko} used to eliminate the honest condition of single-prover quantum zero-knowledge systems depending on the rewinding by Watrous \cite{W2}. We give the protocol in Figure \ref{3gz}, which is transformed from Figure \ref{3turn2} similarly to the transformation from Figure \ref{2turn} to Figure \ref{2gz}, is zero-knowledge even to the malicious verifier by a direct construction of a simulator. (The difference between Figure \ref{3turn2} and Figure \ref{2turn} is the number of turns and the number of provers.) The simulator is described in Figure \ref{3sim}.
\begin{proof}[Proof of Zero Knowledge (sketch)]
Similarly to the proof of Theorem 3, we can assume that the verifier sends the same $b=0$ or $1$ to all provers due to the GHZ test.

After that, the analysis is almost the same as the one by Kobayashi \cite{Ko}, who used the rewinding \cite{W2}.
Denote the malicious verifier's first circuit by $V'_1$. We can denote the state in $(\mathsf{V,M,C})$ by $\Psi_V(x,2)=\frac{1}{2}(|0\rangle\langle0|_C\otimes\sigma_0+|1\rangle\langle1|_C\otimes\sigma_1)$. Here, $\sigma_0,\sigma_1$ are density operators of $\mathsf{V,M}$ and $|0\rangle\langle0|_{\mathsf C}, |1\rangle\langle1|_{\mathsf C}$ are one qubit of $\mathsf C$. (the amplitudes of $|0\rangle\langle0|,|1\rangle\langle1|$ may not be exactly equal since we do not assume perfect zero-knowledge, but the difference is negligible since we assume computational quantum zero-knowledge.) In addition, $\rm Tr_M\sigma_0=\rm Tr_M\sigma_1$ holds. If the output of the measurement of $\mathsf D$ is $0$, then the state in $(\mathsf{I,A,V,M})$ is indistinguishable from the state of the malicious verifier after the third turn. The probability that the output of the measurement of $\mathsf{D}=0,1$ is almost equally $1/2$ since $\mathsf{C}=0,1$ is independent from $\mathsf{C'}$. Hence the rewinding \cite{Ko,W2} works.
\end{proof}
\begin{figure}
\hrulefill\\
Let $V_1,V_2$ be the circuit that the verifier of a three-turn protocol for $\mathsf{HVQMZK^*}(p,3,1-\epsilon, 1-\delta)$ applies after receiving message registers $(\mathsf{M}_1,\mathsf{M}_2,...,\mathsf{M}_p)$ from provers $1,...,p$, respectively, at the first turn and the third turn. The public coin protocol is the following;
\begin{enumerate}
\item Receive $\mathsf V$ from prover 1.
\item Select $b\in\{0,1\}$ uniformly at random.  Send $b$ to all provers.
\item Do the following tests based on $b=0,1$;
\begin{enumerate}
\item ($b=0$: Forward test)\\
Receive $\mathsf {M}_i$ from prover $i$, for $1 \le i \le p$. Apply $V_2$ to $\mathsf{(V,M_1 ,...,M_p)}$. If the state in $\mathsf{(V,M_1 ,...,M_k )}$ is accepted by the original protocol, then accept. Reject otherwise.
\item ($b=1$: Backward test)\\
 Receive $\mathsf {M}_i$ from prover $i$, for $1 \le i \le p$. Apply $(V_1)^\dagger$ to $\mathsf{(V,M_1 ,...,M_p)}$. If all qubits in $\mathsf{(V)}$ are $|0\rangle$, then accept. Reject otherwise.
\end{enumerate}
\end{enumerate}
\hrulefill
\caption{Public coin protocol for a three-turn quantum interactive proof.}
\label{3turn2}
\end{figure}
\begin{figure}
\hrulefill\\
$\mathsf{V,M_i}$ are the same as Figure \ref{3turn2}. In addition, we use $\mathsf{G}_0, \mathsf{G}_1, \mathsf{G}_2,...,\mathsf{G}_p$, which are all 1 qubit register.
At the start, prover $i$ has $\mathsf{G}_i$.
\begin{enumerate}
\item Receive $\mathsf V$ from prover 1.
\item Select $b'\in\{0,1\}$ uniformly at random. Send $b'$  to provers $1,...,p$.
\item Do the following steps depending on $b'=0,1$;
\begin{enumerate}
\item ($b'=0$: GHZ test)\\
Prover $i$ sends $\mathsf{G}_i$, for $i=1,...,p$. Measure ($\mathsf{G}_0, \mathsf{G}_1, \mathsf{G}_2,...,\mathsf{G}_p$) by  the projection onto $\frac{1}{\sqrt2} (|0^{p+1}\rangle+|1^{p+1}\rangle)$. If $\frac{1}{\sqrt2} (|0^{p+1}\rangle+|1^{p+1}\rangle)$ is measured, then accept.
\item ($b'=1$: History test)\\
Prover $i$ sends $\mathsf {M}_i$ and $\mathsf{G}_i$, for $i=1,...,p$. Measure $\mathsf{G}_0$ in the computational basis. Let $b$ the outcome. The verifier simulates step 3 of the protocol in Figure \ref{3turn2} depending on $b$.
\end{enumerate}
\end{enumerate}
\hrulefill
\caption{General zero-knowledge protocol based on the public coin protocol in Figure 8.}
\label{3gz}
\end{figure}
\begin{figure}
\hrulefill\\
Registers used by the simulator.\\\\
$\mathsf{I}$ Input that the malicious verifier uses to get information.\\
$\mathsf{A'}$ Ancillas of the malicious verifier.\\
$\mathsf{C'}$ Coin qubit of the malicious verifier.\\
$\mathsf{V}$ Register that the (honest or malicious) verifier receives at the first turn.\\
$\mathsf{M}$ Register that the (honest or malicious) verifier receives at the third turn.\\
$\mathsf{A}$ Ancillas of the simulator.\\
$\mathsf{C}$ Coin qubit of the honest verifier.\\
$\mathsf{D}$ Register for checking whether $\mathsf{C,C'}$ equals or not.\\
\hrulefill\\
Construction of the simulator using the simulator for the honest verifier at the third turn $S_V$.\\\\
0. Prepare the input $\rho$ in $\mathsf{I}$ and $|0\rangle$ in other registers.\\
1. Apply $S_V$ to $(\mathsf{V,M,C,A})$.\\
2. Apply $ V'_1$ to $(\mathsf{I,A',V,C}')$.\\
3. Copy the XOR of $(\mathsf{C,C'})$ to $\mathsf D$.\\
4. Measure $\mathsf{D}$ in the computational basis. If the measurement outputs $0$, then output $(\mathsf{I,A',C',V,M})$. If the measurement outputs $1$, then apply ${V'_1}^\dagger$ and then apply ${S_V}^\dagger$ to $(\mathsf{V,M,C,A})$.\\
5. If $(\mathsf{A',C',V,M,A,C,D})$ are all $|0\rangle$, then apply the phase flip. Apply $S_V$ and then apply $V'_1$ . Output $(\mathsf{I,A,'C,'V,M})$.\\
\hrulefill
\caption{Simulator for the malicious verifier.}
\label{3sim}
\end{figure}

\end{document}